\documentclass[journal,10pt]{IEEEtran}

\pdfoutput=1

\usepackage{amsfonts,amssymb,amsmath}
\usepackage{hyperref}
\usepackage{graphicx}
\usepackage{caption}
\usepackage{embrac}
\newtheorem{theorem}{Theorem}
\newtheorem{proof}{Proof}
\usepackage{subeqnarray}
\usepackage{cases}
\usepackage{enumerate}
\usepackage{multirow}
\usepackage{multicol}
\usepackage{makecell}

\usepackage{algorithm}
\usepackage{algpseudocode}

\begin{document}

\title{Opportunistic Routing Aided Cooperative Communication Network with Energy Harvesting}

\author{Wannian An, Chen Dong*, Xiaodong Xu, Chao Xu, Shujun Han, and~Lei Teng 


\thanks{This work was supported in part by $\sim$. \emph{(Corresponding author: Chen Dong.)}

Wannian An, Chen Dong, Xiaodong Xu, Shujun Han and Lei Teng are with the State Key Laboratory of Networking and Switching Technology, Beijing University of Posts and Telecommunications, Beijing, 100876, China. (e-mail: anwannian2021@bupt.edu.cn; dongchen@bupt.edu.cn; xuxiaodong@bupt.edu.cn; hanshujun@bupt.edu.cn; tenglei@bupt.edu.cn;).

Chao Xu is with the School of Electronics and Computer Science, University of Southampton, SO17 1BJ Southampton, U.K. (e-mail: cx1g08@ecs.soton.ac.uk).}
}

{}

\maketitle

\begin{abstract}
In this paper, a cooperative communication network based on energy-harvesting (EH) decode-and-forward (DF) relays that harvest energy from the ambience using buffers with harvest-store-use (HSU) architecture is considered. An opportunistic routing (OR) protocol, which selects the transmission path of packet based on the node transmission priority, is proposed to improve data delivery in this network. Additionally, an algorithm based on state transition matrix (STM) is proposed to obtain the probability distribution of the candidate broadcast node set. Based on the probability distribution, the existence conditions and the theoretical expressions for the limiting distribution of energy in energy buffers using discrete-time continuous-state space Markov chain (DCSMC) model are derived. Furthermore, the closed-form expressions for network outage probability and throughput are obtained with the help of the limiting distributions of energy stored in buffers. Numerous experiments have been performed to validate the derived theoretical expressions.
\end{abstract}


\IEEEpeerreviewmaketitle

\section{Introduction}
\IEEEPARstart{I}{n} recent years, as a promising technology, energy-harvesting (EH) has drawn researchers' substantial attention due to its capability of harvesting energy from the surrounding ambient energy sources such as light energy, thermal energy and radio frequency (RF) energy, etc. \cite{1}. Specifically, EH is widely employed in cooperative wireless communication to prolong the lifetime of traditional energy-constrained cooperative wireless communication networks \cite{2}.

The relay selection method considering both the energy state of relay buffer and channel state information (CSI) for EH wireless body area network is shown in \cite{3}. In \cite{4} and \cite{5}, the power splitting protocol is employed at the EH relay of the EH-based cooperative communication network to obtain a tradeoff between the transmission energy and decoding energy. In these studies mentioned above, the energy is harvested with the harvest-use (HU) structure, which means that the harvested energy will be used immediately.

Another efficient method of harvesting energy from the environment is to use the harvest-store-use (HSU) architecture, which allows the harvested energy to be stored in an energy buffer or super-capacitor for later use. A wireless powered communication system with an EH node using HSU architecture, which harvests energy from the RF signals in the downlink and uses the stored energy to transmit data in the uplink, is considered in \cite{6}, \cite{34} and \cite{7}, where the limiting distribution of stored energy in energy buffer at the EH node is derived using discrete-time continuous-state space Markov chain (DCSMC) model. The online and offline optimization algorithms for joint relay selection and power control have been discussed in \cite{8}, which aim to maximize the end-to-end system throughput under the constraints of data and energy storage. In \cite{9}, an EH-based two-way relaying network is studied, where the relay uses the energy harvested from ambient RF signals to drastically reduce the battery energy consumption. In \cite{10}, two time switching policies of the EH relay equipped with energy and data buffers are proposed to maximize the throughput of the cooperative wireless network. To improve the performance of the energy-constrained cooperative communication network, \cite{11} proposes a relay selection scheme based on the status of both energy buffer and data buffer. A cooperative cognitive radio network with two Internet of Things (IoT) devices serving as the relays has been investigated in \cite{12}, where the IoT devices employ a time-splitting-based approach for harvesting energy from the RF signals received from a pair of primary users and processing information. Furthermore, the exact expressions of outage probability for the IoT cooperative communication system under Nakagami-m fading is derived.

In these works mentioned above, when the path loss in practical communication link is high, it is not efficient for the relays to harvest energy from RF signals received from the source nodes \cite{13}. Another strategy with greater practical interest is to build the self-sustaining node (SSN) which harvests energy from the ambience. More explicitly, the two-hop relay networks with an HSU-architecture-based EH relay are considered in \cite{31}, \cite{14} and \cite{15}, where the theoretical expressions for the limiting distribution of energy stored in buffer employing best-effort and on-off policies are derived to analyze the system outage probability and throughput. Different from \cite{14}, a feedback strategy is introduced into the system considered in \cite{15}, which further improves the spectral efficiency of the system. Furthermore, \cite{33} and \cite{16} introduce the two-hop relay networks with a self-sustaining source and a self-sustaining relay, where both the source and the relay harvest energy from the ambience. In \cite{32} and \cite{17}, the performance of a two-hop relay network with a self-sustaining source harvesting energy from the ambience and a data buffer-aided relay is analyzed, where three simple link selection schemes based on buffer status or channel knowledge are adopted to maximize the throughput of the relay network.

It can be found that the above studies are mainly based on a wireless two-hop relay network with relays with the same priority (i.e. they cannot communicate with each other). However, the practical wireless cooperative network often utilizes the wireless multi-hop network with distributed topology, such as the sensor network \cite{18}. For the sake of further improving performance, the wireless multi-hop network based on opportunistic routing (OR) protocol has been widely studied. Specifically, an efﬁcient OR protocol based on both cross-layer information exchange and energy consumption in an ad hoc network has been studied in \cite{19} and \cite{20}. Additionally, the OR protocol considering both global optimization and local optimization is proposed for the dependent duty-cycled wireless sensor network to minimize the end-to-end latency \cite{21}. In order to realize the tradeoff between routing efficiency and computational complexity, multi-objective optimization and Pareto optimality are introduced into the power control-based OR for wireless ad hoc networks \cite{22}. In \cite{23}, the OR utilizing the network-based candidate forwarding set optimization scheme is introduced to reduce the transmission delay and avoid duplicate transmission in wireless multi-hop networks. In order to reduce the energy consumption and improve the data delivery ratio in underwater wireless sensor networks, Coutinho and  Boukerche \cite{24} design two candidate set selection heuristics to jointly select the most suitable acoustic modem and next-hop forwarder candidate nodes for the current hop. Furthermore, in \cite{25}, a reliable reinforcement learning-based OR is proposed in the underwater acoustic sensor network.

According to the above researches, it can be easily found that although EH and OR have been widely studied in cooperative wireless networks, less attention has been devoted to the joint application of EH and OR in cooperative wireless networks. This article aims to study OR-aided cooperative communication networks with EH. Table \ref{tab0} shows the comparison between our work and the above references. More specifically, decode-and-forward (DF) relays are powered by harvested energy from the ambience using the infinite-size buffers with HSU architecture. Additionally, the DCSMC model is used to model energy buffers. The main contributions of this work are as follows:
\begin{enumerate}
  \item In this paper, a cooperative communication network based on EH DF relays is considered. In order to improve the data delivery in this network, an OR protocol for selecting the packet transmission path based on the node transmission priority is proposed.
  \item An algorithm for finding the probability distribution of the candidate broadcast node (CBN) set is proposed based on the STM.
  \item Based on the DCSMC model and the probability distribution of the CBN set, the existence conditions and the theoretical expressions for the limiting distribution of energy in energy buffers are derived.
  \item Based on the limiting distributions of energy stored in buffers, the closed-form expressions are derived for system outage probability and throughput. Additionally, the derived analytical closed-forms and the theoretical analysis given are numerically validated.
\end{enumerate}

\begin{table}
\centering
\caption{Comparison of references. Where N indicates that the technology is not adopted in the studied system, and Y indicates that the technology is adopted in the studied system.}
\label{table}
\renewcommand{\arraystretch}{1.5}
\begin{tabular}{|c|c|c|c|}
\hline
Reference & EH buffer architecture & SSN & OR\\
\hline
\cite{3}, \cite{4}, \cite{5} & HU & N & N\\
\hline
\begin{tabular}[c]{@{}c@{}}\cite{6}, \cite{34}, \cite{7}, \cite{8}, \\
\cite{9}, \cite{10}, \cite{11}, \cite{12} \end{tabular} & HSU & N & N\\
\hline
\begin{tabular}[c]{@{}c@{}}\cite{31}, \cite{14}, \cite{15}, \\
\cite{33}, \cite{16}, \cite{32}, \cite{17} \end{tabular} & HSU & Y & N\\
\hline
\begin{tabular}[c]{@{}c@{}}\cite{19}, \cite{20}, \cite{21}, \\
\cite{22}, \cite{23}, \cite{24}, \cite{25} \end{tabular} & No EH & N & Y\\
\hline
Our work & HSU & Y & Y\\
\hline
\end{tabular}
\label{tab0}
\end{table}

The remainder of this paper is organized as follows. In section II, the system model and the proposed OR protocol is described. In section III, the limiting distributions of energy stored in buffers are shown. In section IV, the expressions are derived for system outage probability and throughput. Simulation and theory performance results are presented in Section V. Finally, the conclusion is presented in Section VI.

$Notations:$ $A\sim\mathcal{CN}(0, \theta)$ indicates the random variable $A$ follows the complex Gaussian distribution with mean 0 and variance $\theta$. The absolute value of $B$ is denoted by $|B|$. $\mathbb{E}[\cdot]$ denotes the expectation operator. $C_1 \cap C_2$ denotes that both the condition $C_1$ and the condition $C_2$ are satisfied, while $C_1 \cup C_2$ means that at least one of condition $C_1$ and condition $C_2$ is satisfied. $W(\cdot)$ is the Lambert W function. Boldface capital and lower-case letters stand for matrices and vectors, respectively. In addition, $\overline{C_1}$ represents the opposite of condition $C_1$. $||\cdot||_{2}$ denotes the Euclidean norm of vector.

\section{System Model and OR Protocol}

\subsection{System Model}

As illustrated in Fig. \ref{fig1}, the network considered in this paper consists of a source node $S$,  a destination node $D$, and two DF relay nodes $R1$ and $R2$. In this network, all nodes operate in half duplex mode. The fixed power supply is assumed in both $S$ and $D$, while $R1$ and $R2$ are equipped energy buffer shown in Fig. \ref{fig2} and powered solely by the energy harvested from the ambience. More especially, in each time slot, there is only one activated node transmit signals to corresponding receiving nodes, while the inactive nodes keep silent or receive signals transmitted from the activated node. Additionally, the quasi-static Rayleigh fading channel model is assumed between all nodes, then, the channel coefficients within the $i$-th time slot between $S$ and $D$, $S$ and $R1$, $S$ and $R2$, $R1$ and $R2$, $R1$ and $D$, and $R2$ and $D$ can be denoted by $h_{SD}(i)\sim\mathcal{CN}(0, d_{SD}^{-\alpha})$, $h_{SR1}(i)\sim\mathcal{CN}(0, d_{SR1}^{-\alpha})$, $h_{SR2}(i)\sim\mathcal{CN}(0, d_{SR2}^{-\alpha})$, $h_{R1R2}(i)\sim\mathcal{CN}(0, d_{R1R2}^{-\alpha})$, $h_{R1D}(i)\sim\mathcal{CN}(0, d_{R1D}^{-\alpha})$ and $h_{R2D}(i)\sim\mathcal{CN}(0, d_{R2D}^{-\alpha})$, respectively, where $d_{nm}$ denotes the distance between the node $n \in [S, R1, R2]$ and the node $m \in [R1, R2, D]$, and $\alpha$ is the path-loss parameter.
\begin{figure}[H]
\centerline{\includegraphics[width=3.5in]{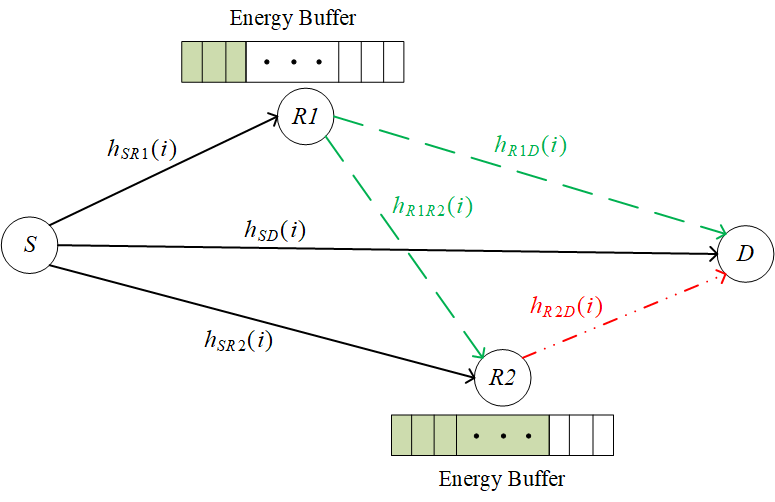}}
\caption{System model.\label{fig1}}
\end{figure}

\begin{figure}[H]
\centerline{\includegraphics[width=3.5in]{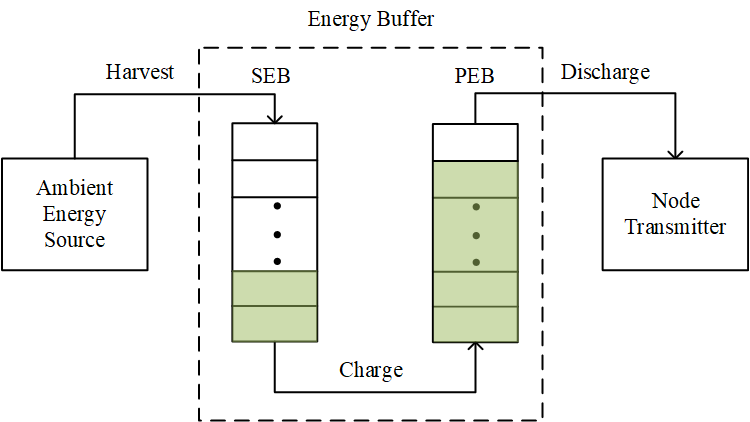}}
\caption{HSU energy harvesting architecture.\label{fig2}}
\end{figure}

The HSU energy harvesting architecture adopted by the energy buffers equipped by $R1$ and $R2$ is depicted in Fig. \ref{fig2}. It can be seen from Fig. \ref{fig2} that the architecture mainly consists of an infinite-size primary energy buffer (PEB) and an infinite-size secondary energy buffer (SEB). Particularly, since the rechargeable energy storage devices are not able to discharge when they are being charged \cite{26}, \cite{27}, the PEB is utilized to power the node transmitter. In contrast, the harvested energy from the ambient energy source (e.g. the predictable solar energy \cite{28}) needs to be stored in the SEB. In addition, for the sake of achieving the target that energy buffers equipped by $R1$ and $R2$ can charge and discharge simultaneously, it is assumed that the PEB may be charged instantaneously by the SEB at the end of one signal time slot. Moreover, the energy loss in the charging process between SEB and PEB and the energy loss in the discharging process between PEB and node transmitter are omitted. In \cite{27}, there is evidence to indicate that these assumptions are feasible.

In the proposed cooperative communication network, as shown in Fig. \ref{fig1}, the $S$ may broadcast unit-energy signals $x_S(i)$ to $R1$, $R2$ and $D$ at rate $R_0$  with the constant transmitting power $P_S$. The received signals $y_{SR1}(i)$, $y_{SR2}(i)$ and $y_{SD}(i)$ at $R1$, $R2$ and $D$ in the $i$-th time slot can be represented by
\begin{equation}\label{eq1}
            y_{SR1}(i) = \sqrt{P_{S}}h_{SR1}(i)x_S(i) + n_{SR1}(i),
\end{equation}
\begin{equation}
            y_{SR2}(i) = \sqrt{P_{S}}h_{SR2}(i)x_S(i) + n_{SR2}(i),
\end{equation}
\begin{equation}
            y_{SD}(i) = \sqrt{P_{S}}h_{SD}(i)x_S(i) + n_{SD}(i),
\end{equation}
where, $n_{SR1}(i)$, $n_{SR2}(i)$ and $n_{SD}(i)\sim\mathcal{CN}(0, N_0)$ denote the received additive white Gaussian noise (AWGN) at $R1$, $R2$ and $D$, respectively. Therefore, the instantaneous link signal to noise ratios (SNRs) $\gamma_{SR1}(i)$, $\gamma_{SR2}(i)$ and $\gamma_{SD}(i)$ at $R1$, $R2$ and $D$ in the $i$-th time slot would be given as follows
\begin{equation}
    \begin{split} 
            & \gamma_{SR1}(i) = \frac{P_S|h_{SR1}(i)|^2}{N_0}, \quad \gamma_{SR2}(i) = \frac{P_S|h_{SR2}(i)|^2}{N_0} \\  
            & \mbox{ and } \quad \gamma_{SD}(i) = \frac{P_S|h_{SD}(i)|^2}{N_0}.   
    \end{split}
\end{equation}

Because of the DF mode, $R1$ decodes the received signals $x_S(i)$, re-encodes the information into unit-energy symbols $x_{R1}(i)$, and then broadcasts $x_{R1}(i)$ to $R2$ and $D$ at rate $R_0$ with the constant power $P_{R1}$. Additionally, $R2$ processes the received signals in the same way as $R1$, and then broadcasts the unit-energy symbols $x_{R2}(i)$ to $D$ at rate $R_0$ with the constant power $P_{R2}$. Next, the received signals $y_{R1R2}(i)$, $y_{R1D}(i)$ and $y_{R2D}(i)$ at $R2$ and $D$ in the $i$-th time slot can be denoted by
\begin{equation}
            y_{R1R2}(i) = \sqrt{P_{R1}}h_{R1R2}(i)x_{R1}(i) + n_{R1R2}(i),
\end{equation}
\begin{equation}
            y_{R1D}(i) = \sqrt{P_{R1}}h_{R1D}(i)x_{R1}(i) + n_{R1D}(i),
\end{equation}
\begin{equation}
            y_{R2D}(i) = \sqrt{P_{R2}}h_{R2D}(i)x_{R2}(i) + n_{R2D}(i),
\end{equation}
where, $n_{R1R2}(i)$, $n_{R1D}(i)$ and $n_{R2D}(i)\sim\mathcal{CN}(0, N_0)$ denote the received AWGN at $R2$ and $D$, respectively. Similarly, the instantaneous link SNRs $\gamma_{R1R2}(i)$, $\gamma_{R1D}(i)$ and $\gamma_{R2D}(i)$ at $R2$ and $D$ can be expressed as follows
\begin{equation}
    \begin{split}
            &\gamma_{R1R2}(i) = \frac{P_{R1}|h_{R1R2}(i)|^2}{N_0}, \quad \gamma_{R1D}(i) = \frac{P_{R1}|h_{R1D}(i)|^2}{N_0} \\
            &\mbox{ and } \quad \gamma_{R2D}(i) = \frac{P_{R2}|h_{R2D}(i)|^2}{N_0}.
    \end{split}
\end{equation}

The time slot diagram of the proposed cooperative communication network is shown in Fig. \ref{fig3}, where one time slot consists of pilot broadcasting sub-slot $t1$, packet broadcasting sub-slot $t2$ and positive acknowledgement (ACK) or negative acknowledgement (NACK) broadcasting sub-slot $t3$. Specifically, in sub-slot $t1$, nodes $D$, $R1$ and $R2$ need to broadcast the pilot signals orderly so that each node of this network can know the CSI between itself and other nodes. In sub-slot $t2$, using the OR protocol described in subsection II B, at most one node is selected to broadcast the packet while the other nodes keep silent. From Fig. \ref{fig3} it can be seen that $S$, $R1$ and $R2$ are selected to broadcast the packet in time slot $T1$, time slot $T2$ and time slot $T3$ respectively. In sub-slot $t3$, nodes $D$, $R1$ and $R2$ broadcast ACK or NACK signals orderly. In particular, when the instantaneous link SNR at the receiving node is not less than the threshold $\Gamma_{th} = 2^{R_0}-1$, ACK signals would be broadcasted by the receiving node as the sign of successful packet reception. Otherwise, NACK signals would be broadcasted by the receiving node as the sign of the packet reception failure.
\begin{figure}[H]
\centerline{\includegraphics[width=3.5in]{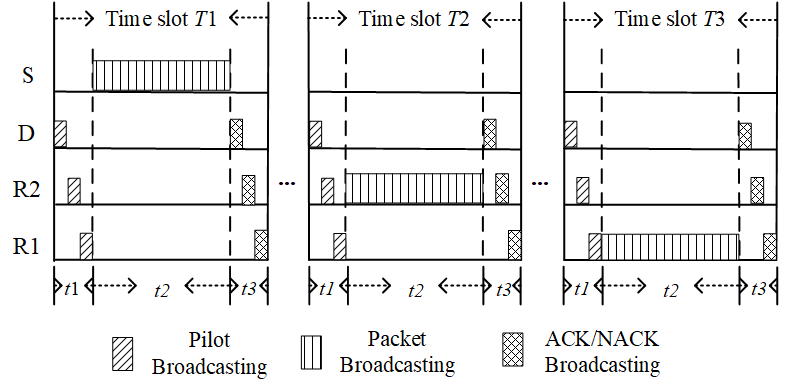}}
\caption{System time slot diagram. Where one time slot consists of pilot broadcasting sub-slot $t1$, packet broadcasting sub-slot $t2$ and positive acknowledgement (ACK) or negative acknowledgement (NACK) broadcasting sub-slot $t3$. Especially, in time slot $T1$, time slot $T2$ and time slot $T3$, $S$, $R1$ and $R2$ are selected to broadcast the packet, respectively.\label{fig3}}
\end{figure}

\subsection{OR Protocol}

\begin{table*}
\centering
\caption{ The OR Protocol. In the current time slot, the protocol determines the broadcast node (BN) $BN(i)$ and effective transmission according to the known candidate broadcast node (CBN) set $\textbf{S}(i)$ and the known conditions (energy information stored by relay nodes and channel information), and further obtains the CBN set $\textbf{S}(i+1)$ of the next time slot.}
\label{table}
\renewcommand{\arraystretch}{1.5}
\begin{tabular}{|c|c|c|c|c|}
\hline
\, CBN set $\textbf{S}(i)$ \, &\qquad Conditions \qquad &\, BN $BN(i)$ \, &\, Effective transmission \, &\, CBN set $\textbf{S}(i+1)$ \,\,\\
\hline
\multirow{5}{*}{$\textbf{s}_1=\{\mbox{S}\}$} & $C_1(i):\gamma_{SD}(i) \ge \Gamma_{th}$ & \multirow{4}{*}{S} & S$\rightarrow$D & $\textbf{s}_1=\{\mbox{S}\}$ \\
           \cline{2-2} \cline{4-5}
           & $C_2(i):\gamma_{SD}(i)<\Gamma_{th}\cap\gamma_{SR1}(i)\ge\Gamma_{th}\cap\gamma_{SR2}(i)<\Gamma_{th}$ &        & S$\rightarrow$R1 & $\textbf{s}_2=\{\mbox{S,R1}\}$ \\
           \cline{2-2} \cline{4-5}
           & $C_3(i):\gamma_{SD}(i)<\Gamma_{th}\cap\gamma_{SR1}(i)<\Gamma_{th}\cap\gamma_{SR2}(i)\ge\Gamma_{th}$ &        & S$\rightarrow$R2 & $\textbf{s}_3=\{\mbox{S,R2}\}$ \\
           \cline{2-2} \cline{4-5}
           & $C_4(i):\gamma_{SD}(i)<\Gamma_{th}\cap\gamma_{SR1}(i)\ge\Gamma_{th}\cap\gamma_{SR2}(i)\ge\Gamma_{th}$ &       & S$\rightarrow$R1,R2 & $\textbf{s}_4=\{\mbox{S,R1,R2}\}$ \\
           \cline{2-5}
           & others & $\simeq$  & $\simeq$ & $\textbf{s}_1=\{\mbox{S}\}$ \\
           \hline

\multirow{6}{*}{$\textbf{s}_2=\{\mbox{S,R1}\}$} & $C_1(i):\gamma_{SD}(i) \ge \Gamma_{th}$ & S & S$\rightarrow$D & \multirow{2}{*}{$\textbf{s}_1=\{\mbox{S}\}$}  \\
           \cline{2-4}
           & $C_5(i):\gamma_{SD}(i)<\Gamma_{th}\cap B_1(i)\ge M_1\cap\gamma_{R1D}(i)\ge\Gamma_{th}$ & R1 & R1$\rightarrow$D & \\
           \cline{2-5}
           & \begin{tabular}[c]{@{}c@{}}$C_6(i):\gamma_{SD}(i)<\Gamma_{th}\cap B_1(i)\ge M_1\cap\gamma_{R1D}(i)<\Gamma_{th}$\\ $\cap\gamma_{SR2}(i)\ge\Gamma_{th}$ \end{tabular} & \multirow{2}{*}{S} & \multirow{2}{*}{S$\rightarrow$R2} & \multirow{3}{*}{$\textbf{s}_4=\{\mbox{S,R1,R2}\}$} \\ \cline{2-2}
           & $C_7(i):\gamma_{SD}(i)<\Gamma_{th}\cap B_1(i)< M_1\cap\gamma_{SR2}(i)\ge\Gamma_{th}$ &  &  & \\
           \cline{2-4}
           & \begin{tabular}[c]{@{}c@{}}$C_8(i):\gamma_{SD}(i)<\Gamma_{th}\cap B_1(i)\ge M_1\cap\gamma_{R1D}(i)<\Gamma_{th}$\\ $\cap\gamma_{SR2}(i)<\Gamma_{th}\cap\gamma_{R1R2}(i)\ge\Gamma_{th}$ \end{tabular} & R1 & R1$\rightarrow$R2 & \\
           \cline{2-5}
           & others & $\simeq$ & $\simeq$ & $\textbf{s}_2=\{\mbox{S,R1}\}$ \\
           \hline

\multirow{3}{*}{$\textbf{s}_3=\{\mbox{S,R2}\}$} & $C_1(i):\gamma_{SD}(i) \ge \Gamma_{th}$ & S & S$\rightarrow$D & \multirow{2}{*}{$\textbf{s}_1=\{\mbox{S}\}$} \\
           \cline{2-4}
           & $C_9(i):\gamma_{SD}(i)<\Gamma_{th}\cap B_2(i)\ge M_2\cap\gamma_{R2D}(i)\ge\Gamma_{th}$ & R2 & R2$\rightarrow$D & \\
           \cline{2-5}
           & others & $\simeq$ & $\simeq$ & $\textbf{s}_3=\{\mbox{S,R2}\}$ \\
           \hline
\multirow{5}{*}{$\textbf{s}_4=\{\mbox{S,R1,R2}\}$} & $C_1(i):\gamma_{SD}(i) \ge \Gamma_{th}$ & S & S$\rightarrow$D & \multirow{4}{*}{$\textbf{s}_1=\{\mbox{S}\}$} \\
           \cline{2-4}
           & $C_9(i):\gamma_{SD}(i)<\Gamma_{th}\cap B_2(i)\ge M_2\cap\gamma_{R2D}(i)\ge\Gamma_{th}$ & R2 & R2$\rightarrow$D & \\
           \cline{2-4}
           & \begin{tabular}[c]{@{}c@{}}$C_{10}(i):\gamma_{SD}(i)<\Gamma_{th}\cap B_2(i)\ge M_2\cap\gamma_{R2D}(i)<\Gamma_{th}$\\ $\cap B_1(i)\ge M_1\cap\gamma_{R1D}(i)\ge\Gamma_{th}$\end{tabular} & \multirow{2}{*}{R1} & \multirow{2}{*}{R1$\rightarrow$D} &  \\
           \cline{2-2}
           & \begin{tabular}[c]{@{}c@{}}$C_{11}(i):\gamma_{SD}(i)<\Gamma_{th}\cap B_2(i)< M_2\cap B_1(i)\ge M_1$\\ $\cap\gamma_{R1D}(i)\ge\Gamma_{th}$\end{tabular} &  &  &  \\
           \cline{2-5}
           & others & $\simeq$ & $\simeq$ & $\textbf{s}_4=\{\mbox{S,R1,R2}\}$ \\
           \hline

\end{tabular}
\label{tab1}
\end{table*}

In this section, the OR protocol is proposed. In the cooperative communication network shown in Fig. \ref{fig1}, it is first assumed that S always has the packet to be transmitted. Secondly, when transmitting the same packet to the same node, $S$ has the highest transmission priority, $R2$ has the second priority, and $R1$ has the lowest priority. This is because $S$ has fixed power support and can transmit the packet to $R1$, $R2$ and $D$. Compared with $S$, EH-based $R1$ can transmit the packet to $R2$ and $D$, while EH-based $R2$ can only transmit the packet to $D$. Third, as long as $D$ receives the packet, $S$ would transmit a new packet in the next time slot. In addition, the transmitting node which has the packet to be forwarded is named as CBN, and the node which can receive the packet transmitted by the CBN is named the neighbouring node of the CBN. The main procedures of OR protocol are divided into the following three steps:
\begin{enumerate}[1) :]
\item Determine the CBN set $\textbf{S}(i) \in\{\textbf{s}_1, \textbf{s}_2, \textbf{s}_3, \textbf{s}_4\}$ in the current time slot, where $\textbf{s}_1 = \{S\}$, $\textbf{s}_2 = \{S,R1\}$, $\textbf{s}_3 = \{S,R2\}$, $\textbf{s}_4 = \{S,R1,R2\}$.
\item According to the channel informations between nodes, the stored energy $B_1(i)$ of $R1$ and the stored energy $B_2(i)$ of $R2$, in the current time slot, the broadcast node (BN) $BN(i)\in\{S, R1, R2, \simeq\}$ is selected from the CBN set $\textbf{S}(i)$, where $\simeq$ indicates that there is no node is selected to broadcast the packet in the current time slot.
\item Determine the effective transmission and get CBN set $\textbf{S}(i+1) \in\{\textbf{s}_1, \textbf{s}_2, \textbf{s}_3, \textbf{s}_4\}$ in the next time slot.
\end{enumerate}

More especially, Table \ref{tab1} describes the OR protocol in detail, where $M_1$ and $M_2$ denote the energy consumed by $R1$ and $R2$ to broadcast the packet in one time slot, respectively. Moreover, $M_1$ and $M_2$ are numerically equal to $P_{R1}$ and $P_{R2}$, respectively.

\section{Limiting Distribution of Energy}
This section is concerned with the limiting distributions of stored energy in the infinite-size energy buffers with HSU architecture. Specially, it is assumed that $X_1(i)$ and $X_2(i)$ are the energy harvested by $R1$ and $R2$ from the surrounding environment in one time slot, respectively. Moreover, $X_1(i)$ and $X_2(i)$ are exponentially distributed random variables (as in both \cite{26} and \cite{27}) with means $\mathbb{E}[X_1(i)]=1/\lambda_1$ and $\mathbb{E}[X_2(i)]=1/\lambda_2$, respectively. So that the probability density function (PDF) of $X_1(i)$ and $X_2(i)$ can be expressed as follows
\begin{equation}
            f_{X_j}(x) = \lambda_j e^{-\lambda_j x}, \quad x>0, j\in{1,2}.
\end{equation}

\subsection{Limiting Distribution of $R2$ Energy}
Using the DCSMC on the continuous state space $[0,\infty)$, the dynamic process $B_2(i)$ of the energy storage in the infinite-size energy buffer of $R2$ can be given as follows
\begin{equation}\label{eq10}
    \begin{split}
             & B_2(i+1) = B_2(i)+X_2(i), \quad \mathbb{P}_{21}\\
             & B_2(i+1) = B_2(i)-M_2+X_2(i), \quad \mathbb{P}_{22}
    \end{split}
\end{equation}
where $\mathbb{P}_{21}$ and $\mathbb{P}_{22}$ are the conditions required for the storage equations in Eq. (\ref{eq10}) to hold according to the OR protocol shown in Table \ref{tab1}. And $\mathbb{P}_{21}$ and $\mathbb{P}_{22}$ can be expressed as follows
\begin{equation}\label{eq11}
    \begin{split}
         \mathbb{P}_{21}:\, & \Big\langle\big( \textbf{S}(i) = \textbf{s}_3 \big) \cup \big(\textbf{S}(i) = \textbf{s}_4 \big)\Big\rangle\cap \bigg[ C_1(i) \cup \Big\langle \overline{C_1}(i)\\
         & \cap\big(B_2(i)\ge M_2\big)\cap\big(\gamma_{R2D}(i)<\Gamma_{th}\big)\Big\rangle \cup \Big\langle\overline{C_1}(i)\\
         & \cap\big( B_2(i)< M_2\big)\Big\rangle\bigg]\cup\big( \textbf{S}(i) = \textbf{s}_1 \big) \cup \big(\textbf{S}(i) = \textbf{s}_2 \big),\\
    \end{split}
\end{equation}
\begin{equation}\label{eq12}
    \begin{split}
             \mathbb{P}_{22}:\, \Big\langle\big( \textbf{S}(i) = \textbf{s}_3 \big) \cup \big(\textbf{S}(i) = \textbf{s}_4 \big)\Big\rangle\cap C_9(i).
    \end{split}
\end{equation}
where $C_1(i)$ and $C_9(i)$ are given in Table \ref{tab1}.

Let $p_1$, $p_2$, $p_3$ and $p_4$ denote the probabilities $\mathrm{Pr}\{\textbf{S}(i) = \textbf{s}_1\}$, $\mathrm{Pr}\{\textbf{S}(i) = \textbf{s}_2\}$, $\mathrm{Pr}\{\textbf{S}(i) = \textbf{s}_3\}$ and $\mathrm{Pr}\{\textbf{S}(i) = \textbf{s}_4\}$ in the case of stable energy buffers, respectively, and the values of these probabilities can be obtained in Alg. \ref{Alg.1} which is introduced in subsection III C. Then, define $\Omega_{SD}=\frac{d_{SD}^{\alpha}N_0}{P_S}$, $\Omega_{SR2}=\frac{d_{SR2}^{\alpha}N_0}{P_S}$, $\Omega_{SR1}=\frac{d_{SR1}^{\alpha}N_0}{P_S}$, $\Omega_{R1D}=\frac{d_{R1D}^{\alpha}N_0}{M_1}$, $\Omega_{R1R2}=\frac{d_{R1R2}^{\alpha}N_0}{M_1}$ and $\Omega_{R2D}=\frac{d_{R2D}^{\alpha}N_0}{M_2}$. Furthermore, let $\psi_2$ denote the energy buffer stabilization parameter for the energy storage process $B_2(i)$ in Eq. (\ref{eq10}), which can be expressed as follows
\begin{equation}\label{eq13}
    \begin{split}
            \psi_2 &=\lambda_2 M_2 \left[\left(p_3+p_4\right)\left(1-e^{-\Omega_{SD}\Gamma_{th}}\right)e^{-\Omega_{R2D}\Gamma_{th}}\right]\\
            & = \frac{M_2 \left[\left(p_3+p_4\right)\left(1-e^{-\Omega_{SD}\Gamma_{th}}\right)e^{-\Omega_{R2D}\Gamma_{th}}\right]}{\mathbb{E}[X_2(i)]}.
    \end{split}
\end{equation}

\begin{theorem}\label{theo1}
If $\psi_2 \leq 1$, the energy storage process $B_2(i)$ in Eq. (\ref{eq10}) will not have a stationary distribution. Moreover, after a finite number of time slots, $B_2(i) > M_2$ almost always hold.
\end{theorem}

\begin{proof}
The proof is presented in Appendix A.
\end{proof}

\begin{theorem}\label{theo2}
For the storage process $B_2(i)$ in Eq. (\ref{eq10}), the limiting energy distribution exists with $\psi_2 > 1$. Furthermore, the limiting PDF of the energy buffer state at $R2$ can be given by
\begin{equation}\label{eq14}
            g_2(x) = \begin{cases}
                          \dfrac{1}{M_2}\left(1-e^{Q_2 x}\right), & 0 \leq x < M_2 \\
                          \dfrac{1}{M_2}\dfrac{-Q_2 e^{Q_2 x}}{\left(b_2\lambda_2+Q_2\right)}, & x \ge M_2
                        \end{cases}
\end{equation}
where,
\begin{equation}\label{eq1400}
    \begin{split}
            b_2 = \left[\left(p_3+p_4\right)\left(1-e^{-\Omega_{SD}\Gamma_{th}}\right)e^{-\Omega_{R2D}\Gamma_{th}}\right],
    \end{split}
\end{equation}
and $Q_2 = \frac{-W\left(-b_2\lambda_2 M_2 e^{-b_2\lambda_2 M_2}\right)}{M_2}-b_2\lambda_2 < 0$, satisfying $b_2 \lambda_2 e^{Q_2 M_2} = b_2 \lambda_2 + Q_2$. Furthermore, the probabilities $\mathrm{Pr}\{B_2(i)\ge M_R\}$ and $\mathrm{Pr}\{B_2(i)< M_R\}$ may be written as follows
\begin{equation}\label{eq141}
    \begin{split}
             \mathrm{Pr}\{B_2(i) \ge M_2\} & =  \int_{x=M_2}^{\infty} \dfrac{-Q_2 e^{Q_2 x}}{M_2\left(b_2\lambda_2 + Q_2\right)}\, dx \\
             & = \dfrac{e^{Q_2 M_2}}{M_2\left(b_2\lambda_2 + Q_2\right)} =\dfrac{1}{M_2 b_2 \lambda_2},
    \end{split}
\end{equation}
\begin{equation}\label{eq142}
    \begin{split}
             \mathrm{Pr}\{B_2(i) < M_2\}  = 1-\dfrac{1}{M_2 b_2 \lambda_2}.
    \end{split}
\end{equation}

\end{theorem}

\begin{proof}
The proof is presented in Appendix B.
\end{proof}

\subsection{Limiting Distribution of $R1$ Energy}
Similarly, Let $B_1(i)$ denotes the dynamic process of the energy storage in the infinite-size energy buffer of relay $R1$ using the DCSMC on a continuous state space $[0,\infty)$, which can be given by the storage equation as follows
\begin{equation}\label{eq15}
    \begin{split}
         & B_1(i+1) = B_1(i)+X_1(i), \quad \mathbb{P}_{11}\\
         & B_1(i+1) = B_1(i)-M_1+X_1(i), \quad \mathbb{P}_{12}
    \end{split}
\end{equation}
Where $\mathbb{P}_{11}$ and $\mathbb{P}_{12}$ are the conditions required for above OR protocol based storage equations to hold, and they are shown in Eq. (\ref{eq16}) and Eq. (\ref{eq17}), respectively.
\begin{figure*} 
	\centering
    \begin{equation}\label{eq16}
        \begin{split}
             \mathbb{P}_{11}:\,&\big(\textbf{S}(i)=\textbf{s}_1\big)\cup\big(\textbf{S}(i)=\textbf{s}_3\big)\cup\big( \textbf{S}(i) = \textbf{s}_2 \big)\cap \Bigg\{ C_1(i)\cup\Big\langle\overline{C_1}(i)\cap\big(B_1(i)\ge M_1\big)\cap\big(\gamma_{R1D}(i)<\Gamma_{th}\big)\Big\rangle\\
             &\cap\bigg[\Big\langle\big(\gamma_{SR2}(i)<\Gamma_{th}\big)\cap\big(\gamma_{R1R2}(i)<\Gamma_{th}\big)\Big\rangle\cup
             \big(\gamma_{SR2}(i)\ge\Gamma_{th}\big)\bigg] \cup\Big\langle\overline{C_1}(i)\cap\big( B_1(i)< M_1\big)\Big\rangle\Bigg\}\\
             &\cup\big(\textbf{S}(i)=\textbf{s}_4\big)\cap \Bigg\{ C_1(i) \cup C_9(i)\cup\overline{C_1}(i)\cap
             \bigg[\Big\langle\big(B_2(i)\ge M_2\big)\cap\big(\gamma_{R2D}(i)<\Gamma_{th}\big)\Big\rangle\cup\big(B_2(i)< M_2\big)\bigg]\\
             &\cap \bigg[\Big\langle\big(B_1(i)\ge M_1\big)\cap\big(\gamma_{R1D}(i)<\Gamma_{th}\big)\Big\rangle\cup\big(B_1(i)< M_1\big)\bigg]\Bigg\}
        \end{split}
    \end{equation}
\end{figure*}
\begin{equation}\label{eq17}
    \begin{split}
        \mathbb{P}_{12}:\,&\bigg[\big(\textbf{S}(i)=\textbf{s}_2\big)\cap\Big\langle C_5(i)\cup C_8(i)\Big\rangle\bigg]\\
        &\cup\bigg[\big(\textbf{S}(i)=\textbf{s}_4\big)\cap\Big\langle C_{10}(i)\cup C_{11}(i)\Big\rangle\bigg].
    \end{split}
\end{equation}
where $C_1(i)$, $C_5(i)$, $C_8(i)$, $C_9(i)$, $C_{10}(i)$ and $C_{11}(i)$ are given in Table \ref{tab1}.

The energy buffer stabilization parameter $\psi_2$ for the energy storage process $B_2(i)$ of Eq. (\ref{eq15}) can be expressed as follows
\begin{equation}\label{eq18}
    \begin{split}
        \psi_1 &=\lambda_1 M_1\left(1-e^{-\Omega_{SD}\Gamma_{th}}\right)\bigg[p_2e^{-\Omega_{R1D}\Gamma_{th}}\\
         & + p_2\left(1-e^{-\Omega_{R1D}\Gamma_{th}}\right)\left(1-e^{-\Omega_{SR2}\Gamma_{th}}\right)e^{-\Omega_{R1R2}\Gamma_{th}} \\
         & + p_4e^{-\Omega_{R1D}\Gamma_{th}}\left(1-\frac{e^{-\Omega_{R2D}\Gamma_{th}}}{M_2 b_2 \lambda_2}\right) \bigg].
    \end{split}
\end{equation}

\begin{theorem}\label{theo3}
Similarly, for the energy storage process $B_1(i)$ in Eq. (\ref{eq15}),  the stationary distribution does not exist when $\psi_1 \leq 1$, and there is always $M_1$ amount of energy available in the buffer.
\end{theorem}

\begin{proof}
The proof of theorem 3 only needs to replace $\mathbb{P}_{21}$ and $\mathbb{P}_{22}$ in Eq. (\ref{APPENDIX A1}) with $\mathbb{P}_{11}$ and $\mathbb{P}_{12}$, respectively, and the rest is the same as the proof process of theorem 1.
\end{proof}

\begin{theorem}\label{theo4}
When, $\psi_1 > 1$, the limiting distribution of $B_1(i)$ exists and the limiting PDF $g_1(x)$ of $B_1(i)$ can be given by
\begin{equation}\label{eq21}
    g_1(x) = \begin{cases}
                  \dfrac{1}{M_1}\left(1-e^{Q_1 x}\right), & 0 \leq x < M_1 \\
                  \dfrac{1}{M_1}\dfrac{-Q_1 e^{Q_1 x}}{\left(b_1 \lambda_1+Q_1 \right)}, & x \ge M_1.
                \end{cases}
\end{equation}
where, $b_1$ is given as follows
\begin{equation}\label{eq211}
    \begin{split}
        b_1 &=\left(1-e^{-\Omega_{SD}\Gamma_{th}}\right)\bigg[p_2e^{-\Omega_{R1D}\Gamma_{th}}\\
         & + p_2\left(1-e^{-\Omega_{R1D}\Gamma_{th}}\right)\left(1-e^{-\Omega_{SR2}\Gamma_{th}}\right)e^{-\Omega_{R1R2}\Gamma_{th}} \\
         & + p_4e^{-\Omega_{R1D}\Gamma_{th}}\left(1-\frac{e^{-\Omega_{R2D}\Gamma_{th}}}{M_2 b_2 \lambda_2}\right) \bigg],
    \end{split}
\end{equation}
and $Q_1 = \frac{-W\left(-b_1\lambda_1 M_1 e^{-b_1\lambda_1 M_1}\right)}{M_1}-b_1\lambda_1 < 0$, satisfying $b_1 \lambda_1 e^{Q_1 M_1} = b_1 \lambda_1 + Q_1$. In addition, the probabilities $\mathrm{Pr}\{B_1(i)\ge M_1\}$ and $\mathrm{Pr}\{B_1(i)< M_1\}$ may be written as follows
\begin{equation}\label{eq212}
    \begin{split}
             \mathrm{Pr}\{B_1(i) \ge M_1\} & =  \int_{x=M_1}^{\infty} \dfrac{-Q_1 e^{Q_1 x}}{M_1\left(b_1\lambda_1 + Q_1\right)}\, dx \\
             & = \dfrac{e^{Q_1 M_1}}{M_1\left(b_1\lambda_1 + Q_1\right)} =\dfrac{1}{M_1 b_1 \lambda_1},
    \end{split}
\end{equation}
\begin{equation}\label{eq213}
    \begin{split}
             \mathrm{Pr}\{B_1(i) < M_1\}  = 1-\dfrac{1}{M_1 b_1 \lambda_1}.
    \end{split}
\end{equation}

\end{theorem}

\begin{proof}
The proof is presented in Appendix C.
\end{proof}

\subsection{The Probability Distribution of CBN Set $\textbf{S}$ in the Case of Stable Energy Buffer}

The previous theorem \ref{theo2} and theorem \ref{theo4} illustrate that after some time slots, the energy buffers in $R1$ and $R2$ will reach their own stable states under the condition of $\psi_1 > 1$ and $\psi_2 > 1$. Moreover, when the energy buffers in $R1$ and $R2$ reach their own stable states, the probability distribution $\textbf{p}= [p_1, p_2, p_3, p_4]$ of CBN Set $\textbf{S}$ may be obtained by STM, which is expressed as follows
\begin{equation}\label{eqC1}
    \textbf{p}(i+1) = \textbf{p}(i)\textbf{T}(i),
\end{equation}
where $\textbf{p}(i)= [p_1(i), p_2(i), p_3(i), p_4(i)]$ denotes the values of $\textbf{p}$ in the $i$-th time slot, $\textbf{p}(i+1)= [p_1(i+1), p_2(i+1), p_3(i+1), p_4(i+1)]$ denotes the values of $\textbf{p}$ in the $(i+1)$-th time slot, and $\textbf{T}(i)$ is the STM of these probabilities from $\textbf{p}(i)$ to $\textbf{p}(i+1)$. Furthermore, according to table \ref{tab1}, the STM $\textbf{T}(i)$ can be constructed as follows
\begin{equation}\label{eqC2}
\textbf{T}(i)=\begin{bmatrix}
        p_{\textbf{s}_1 \textbf{s}_1}(i)      & p_{\textbf{s}_1 \textbf{s}_2}(i)  & p_{\textbf{s}_1 \textbf{s}_3}(i)  & p_{\textbf{s}_1 \textbf{s}_4}(i) \\
        p_{\textbf{s}_2 \textbf{s}_1}(i)      & p_{\textbf{s}_2 \textbf{s}_2}(i)  & p_{\textbf{s}_2 \textbf{s}_3}(i)  & p_{\textbf{s}_2 \textbf{s}_4}(i) \\
        p_{\textbf{s}_3 \textbf{s}_1}(i)      & p_{\textbf{s}_3 \textbf{s}_2}(i)  & p_{\textbf{s}_3 \textbf{s}_3}(i)  & p_{\textbf{s}_3 \textbf{s}_4}(i) \\
        p_{\textbf{s}_4 \textbf{s}_1}(i)      & p_{\textbf{s}_4 \textbf{s}_2}(i)  & p_{\textbf{s}_4 \textbf{s}_3}(i)  & p_{\textbf{s}_4 \textbf{s}_4}(i)
        \end{bmatrix},
\end{equation}
where, the element $p_{IJ}(k)(I,J\in\{\textbf{s}_1, \textbf{s}_2, \textbf{s}_3, \textbf{s}_4\})$ denotes the probability that the CBN set $\textbf{S}(i)$ changes from $I$ in the $i$-th time slot to $J$ in the $(i+1)$-th time slot, satisfying $\begin{matrix} \sum_{J\in\{\textbf{s}_1, \textbf{s}_2, \textbf{s}_3, \textbf{s}_4\}} p_{IJ}(i) \end{matrix}=1$. Furthermore, according to the variables $\gamma_{SD}$, $\gamma_{SR1}$, $\gamma_{SR2}$, $\gamma_{R1D}$, $\gamma_{R1R2}$, $\gamma_{R2D}$, $\mathrm{Pr}\{B_2(i)\ge M_2\}$,  $\mathrm{Pr}\{B_1(i)\ge M_1\}$, and the OR protocol shown in Table \ref{tab1}, $p_{IJ}(i)$ may be obtained as follows
\begin{equation}\label{eqC3}
    \begin{split}
            p_{\textbf{s}_1 \textbf{s}_1}(i) & = \mathrm{Pr}\{\gamma_{SD}(i)<\Gamma_{th}, \gamma_{SR1}(i)<\Gamma_{th}, \gamma_{SR2}(i)<\Gamma_{th}\} \\
            & \quad + \mathrm{Pr}\{C_1(i)\} \\
            & = \left(1-e^{-\Omega_{SD}\Gamma_{th}}\right)\left(1-e^{-\Omega_{SR1}\Gamma_{th}}\right)\left(1-e^{-\Omega_{SR2}\Gamma_{th}}\right) \\
            & \quad + e^{-\Omega_{SD}\Gamma_{th}},
    \end{split}
\end{equation}
\begin{equation}\label{eqC4}
    \begin{split}
            p_{\textbf{s}_1 \textbf{s}_2}(i) & = \mathrm{Pr}\{C_2(i)\} \\
            & = \left(1-e^{-\Omega_{SD}\Gamma_{th}}\right)e^{-\Omega_{SR1}\Gamma_{th}}\left(1-e^{-\Omega_{SR2}\Gamma_{th}}\right),
    \end{split}
\end{equation}
\begin{equation}\label{eqC5}
    \begin{split}
            p_{\textbf{s}_1 \textbf{s}_3}(i) & = \mathrm{Pr}\{C_3(i)\} \\
            & = \left(1-e^{-\Omega_{SD}\Gamma_{th}}\right)\left(1-e^{-\Omega_{SR1}\Gamma_{th}}\right)e^{-\Omega_{SR2}\Gamma_{th}},
    \end{split}
\end{equation}
\begin{equation}\label{eqC500}
    \begin{split}
            p_{\textbf{s}_1 \textbf{s}_4}(i) & = \mathrm{Pr}\{C_4(i)\} \\
            & = \left(1-e^{-\Omega_{SD}\Gamma_{th}}\right)e^{-\Omega_{SR1}\Gamma_{th}}e^{-\Omega_{SR2}\Gamma_{th}},
    \end{split}
\end{equation}
and,
\begin{equation}\label{eqC6}
    \begin{split}
            p_{\textbf{s}_2 \textbf{s}_1}(i) & = \mathrm{Pr}\{C_1(i)\}+\mathrm{Pr}\{C_5(i)\} \\
            & = e^{-\Omega_{SD}\Gamma_{th}} + \frac{\left(1-e^{-\Omega_{SD}\Gamma_{th}}\right)e^{-\Omega_{R1D}\Gamma_{th}}}{b_1(i) \lambda_1 M_1},
    \end{split}
\end{equation}
\begin{equation}\label{eqC11}
    \begin{split}
        p_{\textbf{s}_2 \textbf{s}_2}(i) & = \mathrm{Pr}\{\overline{C_1}(i),B_1(i)\ge M_1, \gamma_{R1D}(i)< \Gamma_{th},\\
        & \quad \gamma_{SR2}(i)< \Gamma_{th}, \gamma_{R1R2}(i)< \Gamma_{th}\} \\
        & \quad + \mathrm{Pr}\{\overline{C_1}(i),B_1(i)< M_1,\gamma_{SR2}(i)< \Gamma_{th}\} \\
        & = \left(1-e^{-\Omega_{SD}\Gamma_{th}}\right)\left(1-e^{-\Omega_{SR2}\Gamma_{th}}\right) \\
        & \quad \times \bigg[\frac{\left(1-e^{-\Omega_{R1D}\Gamma_{th}}\right)\left(1-e^{-\Omega_{R1R2}\Gamma_{th}}\right)}{b_1(i) \lambda_1 M_1}\\
        & \quad + \left(1-\frac{1}{b_1(i) \lambda_1 M_1} \right)\bigg],
    \end{split}
\end{equation}
\begin{equation}\label{eqC12}
    \begin{split}
        p_{\textbf{s}_2 \textbf{s}_4}(i) & = \mathrm{Pr}\{C_6(i)\}+\mathrm{Pr}\{C_7(i)\}+\mathrm{Pr}\{C_8(i)\} \\
        & = \left(1-e^{-\Omega_{SD}\Gamma_{th}}\right)\bigg[ e^{-\Omega_{SR2}\Gamma_{th}} \left( 1-\frac{e^{-\Omega_{R1D}\Gamma_{th}}}{b_1(i) \lambda_1 M_1}\right) \\
        &+\frac{\left(1-e^{-\Omega_{R1D}\Gamma_{th}}\right)\left(1-e^{-\Omega_{SR2}\Gamma_{th}}\right)e^{-\Omega_{R1R2}\Gamma_{th}}}{b_1(i)\lambda_1 M_1} \bigg],
    \end{split}
\end{equation}
\begin{equation}\label{eqC1200}
    \begin{split}
            p_{\textbf{s}_2 \textbf{s}_3}(i) & = 0,
    \end{split}
\end{equation}
and,
\begin{equation}\label{eqC13}
    \begin{split}
            p_{\textbf{s}_3 \textbf{s}_1}(i) & = \mathrm{Pr}\{C_1(i)\}+\mathrm{Pr}\{C_9(i)\} \\
            & = e^{-\Omega_{SD}\Gamma_{th}} + \frac{\left(1-e^{-\Omega_{SD}\Gamma_{th}}\right)e^{-\Omega_{R2D}\Gamma_{th}}}{b_2(i)\lambda_2 M_2},
    \end{split}
\end{equation}
\begin{equation}\label{eqC17}
    \begin{split}
             p_{\textbf{s}_3 \textbf{s}_3}(i) &= \mathrm{Pr}\{\overline{C_1}(i),B_2(i)< M_2 \} \\
            & \quad + \mathrm{Pr}\{\overline{C_1}(i),B_2(i)\ge M_2, \gamma_{R2D}(i) < \Gamma_{th}\} \\
            & = \left(1-e^{-\Omega_{SD}\Gamma_{th}}\right)\left(1-\frac{e^{-\Omega_{R2D}\Gamma_{th}}}{b_2(i)\lambda_2 M_2} \right),
    \end{split}
\end{equation}
\begin{equation}\label{eqC18}
    \begin{split}
            p_{\textbf{s}_3 \textbf{s}_2}(i) = p_{\textbf{s}_3 \textbf{s}_4}(i) = 0,
    \end{split}
\end{equation}
and,
\begin{equation}\label{eqC19}
    \begin{split}
        p_{\textbf{s}_4 \textbf{s}_1}(i) & = \mathrm{Pr}\{C_1(i)\}+\mathrm{Pr}\{C_9(i)\}+\mathrm{Pr}\{C_{10}(i)\} \\
        & \quad +\mathrm{Pr}\{C_{11}(i)\}\\
        & = e^{-\Omega_{SD}\Gamma_{th}} + \frac{\left(1-e^{-\Omega_{SD}\Gamma_{th}}\right)e^{-\Omega_{R2D}\Gamma_{th}}}{b_2(i)\lambda_2 M_2}\\
        & \quad + \frac{\left(1-e^{-\Omega_{SD}\Gamma_{th}}\right)e^{-\Omega_{R1D}\Gamma_{th}}}{b_1(i)\lambda_1 M_1}\\
        & \quad \times \left(1-\frac{e^{-\Omega_{R2D}\Gamma_{th}}}{b_2(i)\lambda_2 M_2}\right),
        \end{split}
\end{equation}
\begin{equation}\label{eqC20}
    \begin{split}
        p_{\textbf{s}_4 \textbf{s}_4}(i) &= \mathrm{Pr}\{\overline{C_1}(i),B_2(i)\ge M_2,\gamma_{R2D}(i) < \Gamma_{th},\\
        & \quad B_1(i)\ge M_1,\gamma_{R1D}(i) < \Gamma_{th} \} \\
        & + \mathrm{Pr}\{\overline{C_1}(i),B_2(i)\ge M_2,\gamma_{R2D}(i) < \Gamma_{th},\\
        & \quad B_1(i)< M_1, \} \\
        & + \mathrm{Pr}\{\overline{C_1}(i),B_2(i)< M_2,B_1(i)\ge M_1,\\
        & \quad \gamma_{R1D}(i) < \Gamma_{th} \} \\
        & + \mathrm{Pr}\{\overline{C_1}(i),B_2(i)< M_2,B_1(i)< M_1 \} \\
        & = 1 - p_{\textbf{s}_4 \textbf{s}_1}(i),
    \end{split}
\end{equation}
\begin{equation}\label{eqC21}
    \begin{split}
            p_{\textbf{s}_4 \textbf{s}_2}(i) = p_{\textbf{s}_4 \textbf{s}_3}(i) = 0,
    \end{split}
\end{equation}
where $C_l(i)(l\in \{1,2,\cdots, 11\})$ is given in Table \ref{tab1}, $b_2(i)$ and $b_1(i)$ are given by
\begin{equation}\label{eqC22}
    \begin{split}
            b_2(i) = \left[\left(p_3(i)+p_4(i)\right)\left(1-e^{-\Omega_{SD}\Gamma_{th}}\right)e^{-\Omega_{R2D}\Gamma_{th}}\right],
    \end{split}
\end{equation}
\begin{equation}\label{eqC23}
    \begin{split}
        b_1(i) &=\left(1-e^{-\Omega_{SD}\Gamma_{th}}\right)\bigg[p_2(i)e^{-\Omega_{R1D}\Gamma_{th}}\\
         & + p_2(i)\left(1-e^{-\Omega_{R1D}\Gamma_{th}}\right)\left(1-e^{-\Omega_{SR2}\Gamma_{th}}\right)e^{-\Omega_{R1R2}\Gamma_{th}} \\
         & + p_4(i)e^{-\Omega_{R1D}\Gamma_{th}}\left(1-\frac{e^{-\Omega_{R2D}\Gamma_{th}}}{M_2 b_2(i) \lambda_2}\right) \bigg].\\
    \end{split}
\end{equation}

The detailed process for obtaining probability distribution $\textbf{p}$ is shown in Alg. \ref{Alg.1}. Specifically, in line 5, the judgment condition $(\textbf{p}_{j_1}(i)\leq 0\,\vert\,\textbf{T}_{j_1,j_2}(i)<0)$ indicates that there are non-positive elements in $\textbf{p}(i)$ or $\textbf{T}(i)$, which is not desirable, so the update process needs to be terminated. Moreover, in line 8, $(||\textbf{p}(i)-\textbf{p}(i+1)||_{2}<10^{-7})$ indicates that the update error is small enough and the update process has converged. Therefore, the update process can be terminated.
\begin{algorithm}
    \caption{Probability Distribution $\textbf{p}$ of CBN Set $\textbf{S}$ based on STM}
    \label{Alg.1}
    \begin{algorithmic}[1]
        \Require $\Omega_{SD}$, $\Omega_{SR1}$, $\Omega_{SR2}$, $\Omega_{R1D}$, $\Omega_{R1R2}$, $\Omega_{R2D}$, $\Gamma_{th}$, $\lambda_1$, $\lambda_2$, $M_1$, $M_2$.
            \State Initialize $i=0$, $k=0$, $j_1$ and $j_2 \in \{1,2,3,4\}$, $\textbf{p}(0)=[p_1(0), p_2(0), p_3(0), p_4(0)]=[\frac{1}{4}, \frac{1}{4}, \frac{1}{4}, \frac{1}{4}]$.
            \While{1}
                \State Calculate $\textbf{T}(i)$ according to $\textbf{p}(i)$ and equations from Eq. (\ref{eqC3}) to Eq. (\ref{eqC23}),
                \State $\textbf{p}(i+1) =\textbf{p}(i)\textbf{T}(i)$,
                \If {$(\textbf{p}_{j_1}(i)\leq 0\,\vert\,\textbf{T}_{j_1,j_2}(i)<0)$}
                \State $k=i-1$,
                \State break,
                \ElsIf {$(||\textbf{p}(i)-\textbf{p}(i+1)||_{2}<10^{-7})$}
                \State $k=i$,
                \State break,
                \Else
                \State $i=i+1$,
                \EndIf
            \EndWhile
            \Ensure Probability Distribution $\textbf{p}(k)$.
    \end{algorithmic}
\end{algorithm}

\section{Analysis of Outage Probability and Throughput}
This section is concerned with the system performance metrics, including outage probability and throughput, when the states of system energy buffers are stable. Specifically, with the OR protocol shown in Table \ref{tab1}, the system is in outage if the destination node $D$ fails to receive the packet from the transmitting node $S$, $R1$ or $R2$. Furthermore, the system outage probability $OP$ is defined as follows
\begin{equation}\label{SEC IV1}
    \begin{split}
             OP = 1-\left(P_{\textbf{s}_1}+P_{\textbf{s}_2}+P_{\textbf{s}_3}+P_{\textbf{s}_4}\right),
    \end{split}
\end{equation}
where, $P_{\textbf{s}_1}$, $P_{\textbf{s}_2}$, $P_{\textbf{s}_3}$ and $P_{\textbf{s}_4}$ represent the probability that the destination node $D$ successfully receive the packet in the case of $\textbf{S}(i)=\textbf{s}_1$, $\textbf{S}(i)=\textbf{s}_2$, $\textbf{S}(i)=\textbf{s}_3$, $\textbf{S}(i)=\textbf{s}_4$, respectively. According to Table \ref{tab1}, they are expressed as follows
\begin{equation}\label{SEC IV2}
    \begin{split}
         P_{\textbf{s}_1} & = \mathrm{Pr}\{\textbf{S}(i)=\textbf{s}_1\}\mathrm{Pr}\{C_1(i)\} \\
         & = p_1 e^{-\Omega_{SD}\Gamma_{th}},
    \end{split}
\end{equation}
\begin{equation}\label{SEC IV3}
    \begin{split}
         P_{\textbf{s}_2} & = \mathrm{Pr}\{\textbf{S}(i)=\textbf{s}_2\}\Big(\mathrm{Pr}\{C_1(i)\}+\mathrm{Pr}\{C_5(i)\}\Big) \\
         & = p_2\Bigg[e^{-\Omega_{SD}\Gamma_{th}}+ \frac{\left(1-e^{-\Omega_{SD}\Gamma_{th}}\right)e^{-\Omega_{R1D}\Gamma_{th}}}{b_1 \lambda_1 M_1} \Bigg],
    \end{split}
\end{equation}
\begin{equation}\label{SEC IV4}
    \begin{split}
         P_{\textbf{s}_3} & = \mathrm{Pr}\{\textbf{S}(i)=\textbf{s}_3\}\Big(\mathrm{Pr}\{C_1(i)\}+\mathrm{Pr}\{C_9(i)\}\Big) \\
         & = p_3\Bigg[e^{-\Omega_{SD}\Gamma_{th}}+ \frac{\left(1-e^{-\Omega_{SD}\Gamma_{th}}\right)e^{-\Omega_{R2D}\Gamma_{th}}}{b_2 \lambda_2 M_2} \Bigg],
    \end{split}
\end{equation}
\begin{equation}\label{SEC IV400}
    \begin{split}
         P_{\textbf{s}_4} & = \mathrm{Pr}\{\textbf{S}(i)=\textbf{s}_4\}\Big(\mathrm{Pr}\{C_1(i)\}+\mathrm{Pr}\{C_9(i)\}\\
         &+ \mathrm{Pr}\{C_{10}(i)\}+\mathrm{Pr}\{C_{11}(i)\}\Big) \\
         & = p_4\Bigg[e^{-\Omega_{SD}\Gamma_{th}}+ \frac{\left(1-e^{-\Omega_{SD}\Gamma_{th}}\right)e^{-\Omega_{R2D}\Gamma_{th}}}{b_2 \lambda_2 M_2}+ \\
         &\frac{\left(1-e^{-\Omega_{SD}\Gamma_{th}}\right)e^{-\Omega_{R1D}\Gamma_{th}}}{b_1 \lambda_1 M_1}\left(1-\frac{e^{-\Omega_{R2D}\Gamma_{th}}}{b_2 \lambda_2 M_2} \right)\Bigg],
    \end{split}
\end{equation}
where $C_1(i)$, $C_5(i)$, $C_9(i)$, $C_{10}(i)$ and $C_{11}(i)$ are given in Table \ref{tab1}. According to $b_2$ in Eq. (\ref{eq1400}) and the equations from Eq. (\ref{SEC IV1}) to Eq. (\ref{SEC IV400}), it can be known that the system outage probability $OP$ in Eq. (\ref{SEC IV1}) can be rewritten as follows
\begin{equation}\label{SEC IV5}
    \begin{split}
             OP  =& 1-\Bigg\{e^{-\Omega_{SD}\Gamma_{th}}+\frac{1}{\lambda_2 M_2}\\
             & +\Bigg[p_2 + p_4\left(1-\frac{e^{-\Omega_{R2D}\Gamma_{th}}}{b_2 \lambda_2 M_2} \right)\Bigg]\\
             & \times\frac{\left(1-e^{-\Omega_{SD}\Gamma_{th}}\right)e^{-\Omega_{R1D}\Gamma_{th}}}{b_1 \lambda_1 M_1}\Bigg\}.
    \end{split}
\end{equation}

Furthermore, considering that the system needs to consume sub slots $t1$ and $t2$ for pilot broadcasting and NACK/ACK broadcasting at the beginning and end of a time slot, as shown in Fig. \ref{fig3}, the throughput $\pi$ of the system is defined as follows
\begin{equation}\label{SEC IV8}
    \begin{split}
         \pi & = \eta R_0\left(1 - OP\right)\\
         & = \Bigg\{e^{-\Omega_{SD}\Gamma_{th}}+\frac{1}{\lambda_2 M_2} + \Bigg[p_2 + p_4\left(1-\frac{e^{-\Omega_{R2D}\Gamma_{th}}}{b_2 \lambda_2 M_2} \right)\Bigg]\\
         & \quad \times\frac{\left(1-e^{-\Omega_{SD}\Gamma_{th}}\right)e^{-\Omega_{R1D}\Gamma_{th}}}{b_1 \lambda_1 M_1}\Bigg\}\eta R_0,
    \end{split}
\end{equation}
Where $\eta$ is the loss factor of throughput $\pi$. In addition, in order to analyze the impact of data transmission rate $R_0$ on system throughput $\pi$, the derivative of $\pi$ with respect to $R_0$ is shown in Eq. (\ref{SEC IV10}) and Eq. (\ref{SEC IV11}). Furthermore, the maximum throughput $\pi_{max}$ and the optimum value $R^{max}_0$ of data transmission rate $R_0$ may be obtained by making the derivative $\frac{d\pi}{d R_0} = 0$.
\begin{figure*} 
	\centering
    \begin{equation}\label{SEC IV10}
        \begin{split}
            \dfrac{d\pi}{d R_0} &= \dfrac{d\{\eta R_0 \left(1 - OP\right)\}}{d R_0} = \eta \left(1 - OP\right) -\eta R_0 \dfrac{dOP}{d R_0} \\
            & = \Bigg\{e^{-\Omega_{SD}\Gamma_{th}}+\frac{1}{\lambda_2 M_2} + \bigg[p_2 + p_4\left(1-\frac{e^{-\Omega_{R2D}\Gamma_{th}}}{b_2 \lambda_2 M_2} \right)\bigg]\frac{\left(1-e^{-\Omega_{SD}\Gamma_{th}}\right)e^{-\Omega_{R1D}\Gamma_{th}}}{b_1 \lambda_1 M_1}\Bigg\}\eta + \frac{p_4 2^{R_{0}}\ln(2)\eta R_0 \Omega_{SD}}{b_1 b_2 \lambda_1 \lambda_2 M_1 M_2}\\
            & \times e^{-\left(\Omega_{SD}+\Omega_{R1D}+\Omega_{R2D}\right)\Gamma_{th}}- 2^{R_{0}}\ln(2)\eta R_0 \Omega_{SD}e^{-\Omega_{R1D}\Gamma_{th}} + \eta R_0 \bigg[p_2 + p_4\left(1-\frac{e^{-\Omega_{R2D}\Gamma_{th}}}{b_2 \lambda_2 M_2} \right)\bigg]\\
            &\times \Bigg\{\frac{2^{R_{0}}\ln(2) e^{-\Omega_{R1D}\Gamma_{th}}}{b_1 \lambda_1 M_1}\bigg[e^{-\Omega_{R1D}\Gamma_{th}}\left(\Omega_{SD}+\Omega_{R1D}\right)-\Omega_{R1D}\bigg] -\frac{\left(1-e^{-\Omega_{SD}\Gamma_{th}}\right)e^{-\Omega_{R1D}\Gamma_{th}}}{{b_1}^{2} \lambda_1 M_1}\dfrac{d b_1}{d R_0} \Bigg\},
        \end{split}
    \end{equation}
\end{figure*}
\begin{figure*} 
	\centering
    \begin{equation}\label{SEC IV11}
        \begin{split}
            \dfrac{d b_1}{d R_0} & = p_2 2^{R_{0}}\ln(2)\Omega_{SD}e^{-\Omega_{SD}\Gamma_{th}}\bigg[e^{-\Omega_{R1D}\Gamma_{th}}+\left(1-e^{-\Omega_{R1D}\Gamma_{th}}\right)
            \left(1-e^{-\Omega_{SR2}\Gamma_{th}}\right)e^{-\Omega_{R1R2}\Gamma_{th}}\bigg]+p_2\left(1-e^{-\Omega_{SD}\Gamma_{th}}\right)\\
            &\times\Bigg\{2^{R_{0}}\ln(2)\Omega_{R1D}e^{-\Omega_{R1D}\Gamma_{th}}\bigg[\left(1-e^{-\Omega_{SR2}\Gamma_{th}}\right)e^{-\Omega_{R1R2}\Gamma_{th}}-1 \bigg] + 2^{R_{0}}\ln(2)\left(1-e^{-\Omega_{R1D}\Gamma_{th}}\right)e^{-\Omega_{R1R2}\Gamma_{th}}\\
            &\times \bigg[\left(\Omega_{SR2}+\Omega_{R1R2}\right) e^{-\Omega_{SR2}\Gamma_{th}}-\Omega_{R1R2}\bigg]\Bigg\} + p_4 2^{R_{0}}\ln(2)e^{-\Omega_{R1D}\Gamma_{th}}\bigg[\left(\Omega_{SD}+\Omega_{R1D}\right) e^{-\Omega_{SD}\Gamma_{th}}-\Omega_{R1D}\bigg]\\
            & \times \left(1-\frac{e^{-\Omega_{R2D}\Gamma_{th}}}{b_2 \lambda_2 M_2}\right) + \frac{p_4 2^{R_{0}}\ln(2) \Omega_{SD}e^{-\left(\Omega_{SD}+\Omega_{R1D}+\Omega_{R2D}\right)\Gamma_{th}}}{b_2 \lambda_2 M_2}.
        \end{split}
    \end{equation}
\end{figure*}

\section{Numerical Results}
In this section, the derived analytical closed-forms and the theoretical analysis given will be numerically validated in this section. For all simulations, the system parameters are considered when the stability conditions $\psi_1>1$ and $\psi_2>1$ are satisfied. Specifically, assume $S$, $R1$, $R2$ and $D$ are all located in a two dimensional plane, and their position coordinates in meters are (0, 0), (30, 20), (60, -20) and (100, 0), respectively. Moreover, let the noise variance $N_0 = -50$ dBm, the path-loss parameter $\alpha=3$, and the loss factor $\eta=5\%$. Besides, in all the figures except Fig. \ref{fig14}, markers denote simulation values and lines represent the STM-based theoretical values.

\begin{figure}[H]
\centerline{\includegraphics[width=3.5in]{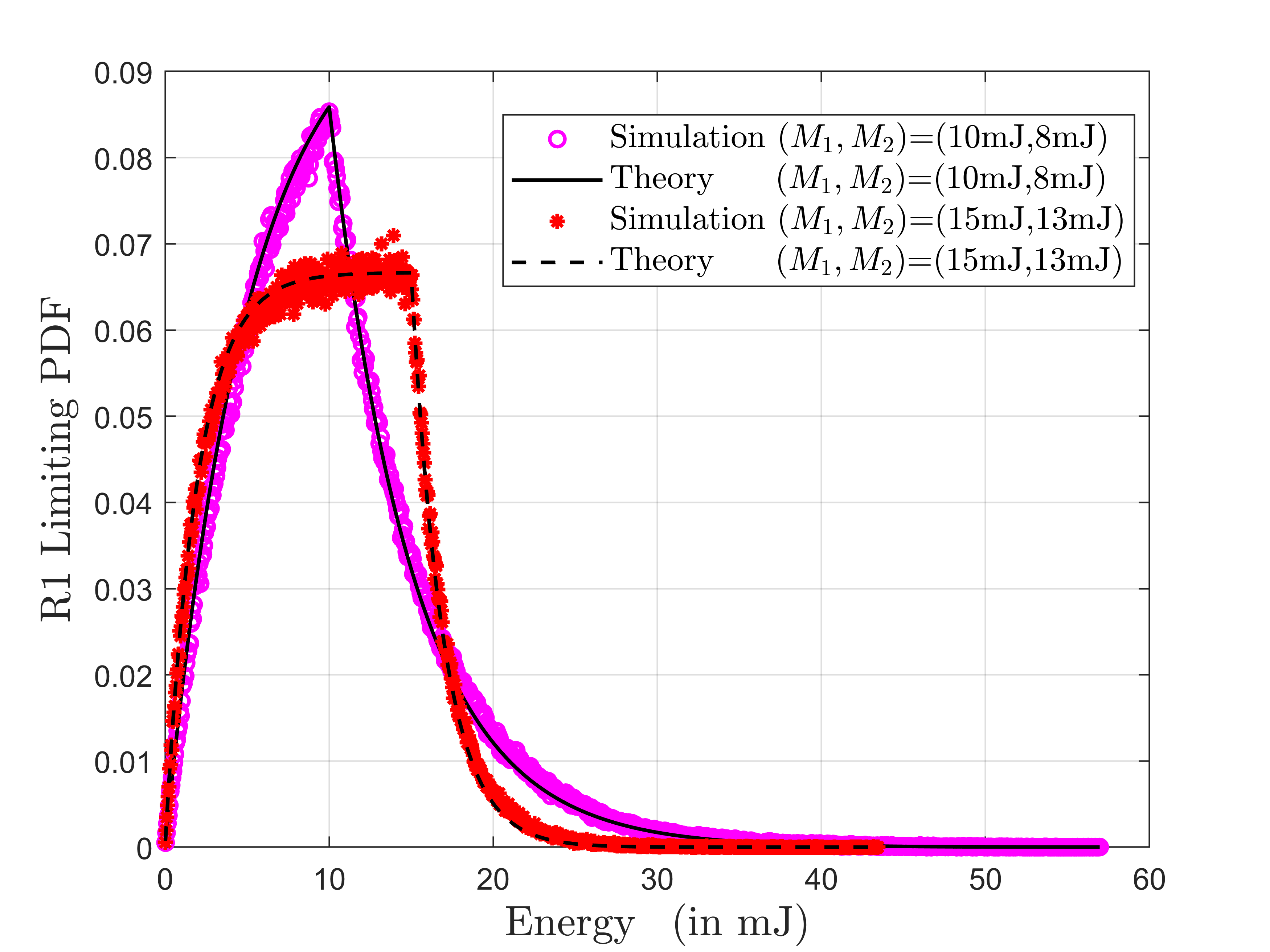}}
\caption{Limiting PDF of energy stored in $R1$ buffer, with parameters $\frac{1}{\lambda_1} = -6$ dB, $\frac{1}{\lambda_2} = -8$ dB, $R_0 = 3$ bit/s/Hz, $P_S = 15$ dBm and two different relay energy consumptions \{$(M_1, M_2)$ = (10 mJ, 8 mJ), (15 mJ, 13 mJ)\}.\label{fig4}}
\end{figure}

\begin{figure}[H]
\centerline{\includegraphics[width=3.5in]{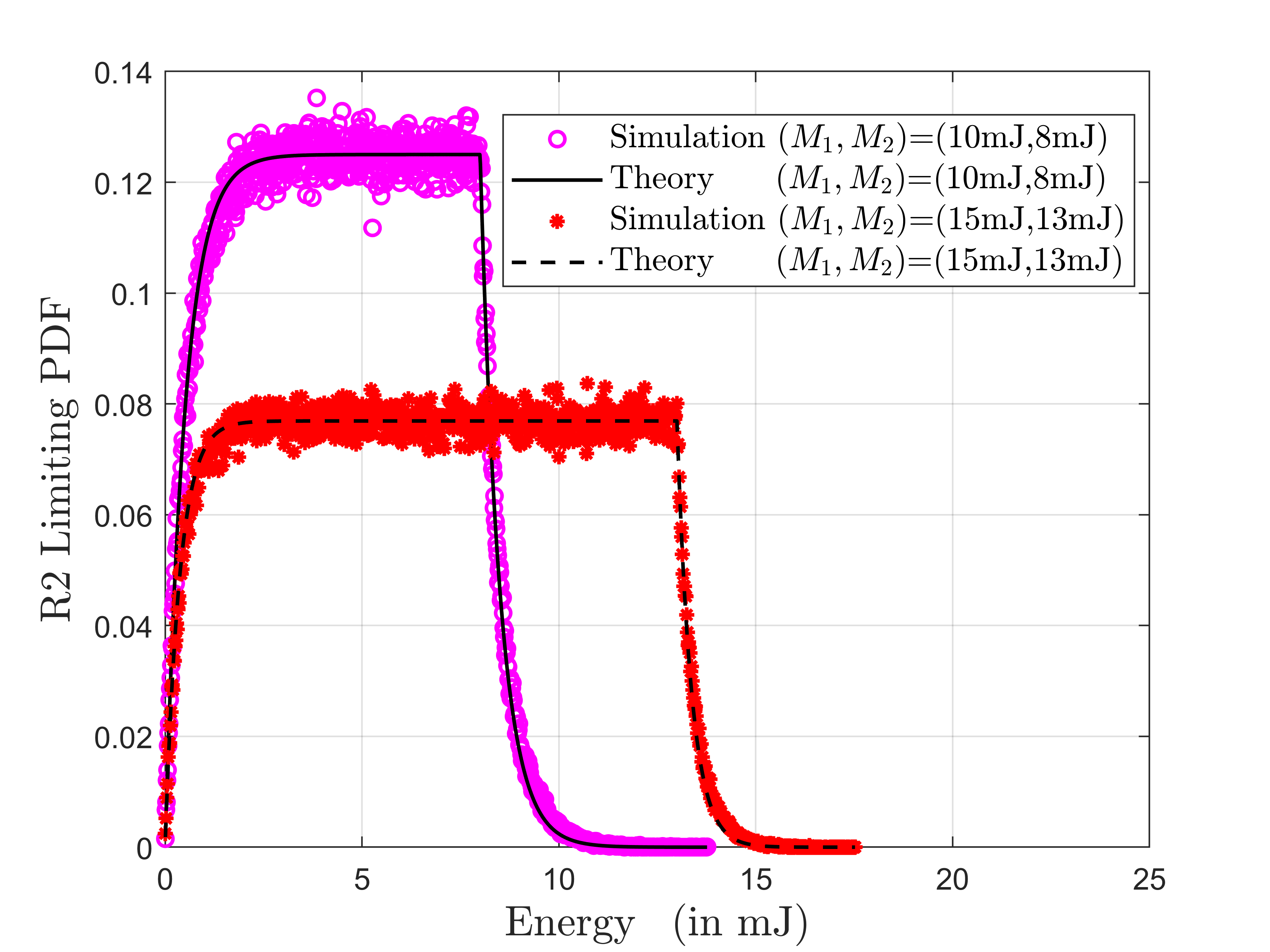}}
\caption{Limiting PDF of energy stored in $R2$ buffer, with parameters $\frac{1}{\lambda_1} = -6$ dB, $\frac{1}{\lambda_2} = -8$ dB, $R_0 = 3$ bit/s/Hz, $P_S = 15$ dBm, and two different relay energy consumptions \{$(M_1, M_2)$ = (10 mJ,8 mJ), (15 mJ, 13 mJ)\}.\label{fig5}}
\end{figure}

Fig. \ref{fig4} and Fig. \ref{fig5} depict the limiting PDF of energy stored in $R1$ buffer and $R2$ buffer when the relay energy consumption is $(M_1, M_2)$ = (10 mJ, 8 mJ) and $(M_1, M_2)$ = (15 mJ, 13 mJ), respectively. It can be clearly seen from Fig. \ref{fig4} that the theoretical PDF curve matchs the simulated PDF scatters. Furthermore, the same case can be seen from Fig. \ref{fig5}. Both Fig. \ref{fig4} and Fig. \ref{fig5} effectively verifies the derived theoretical expressions in Eq. (\ref{eq14}) and Eq. (\ref{eq21}).

\begin{figure}
\centerline{\includegraphics[width=3.5in]{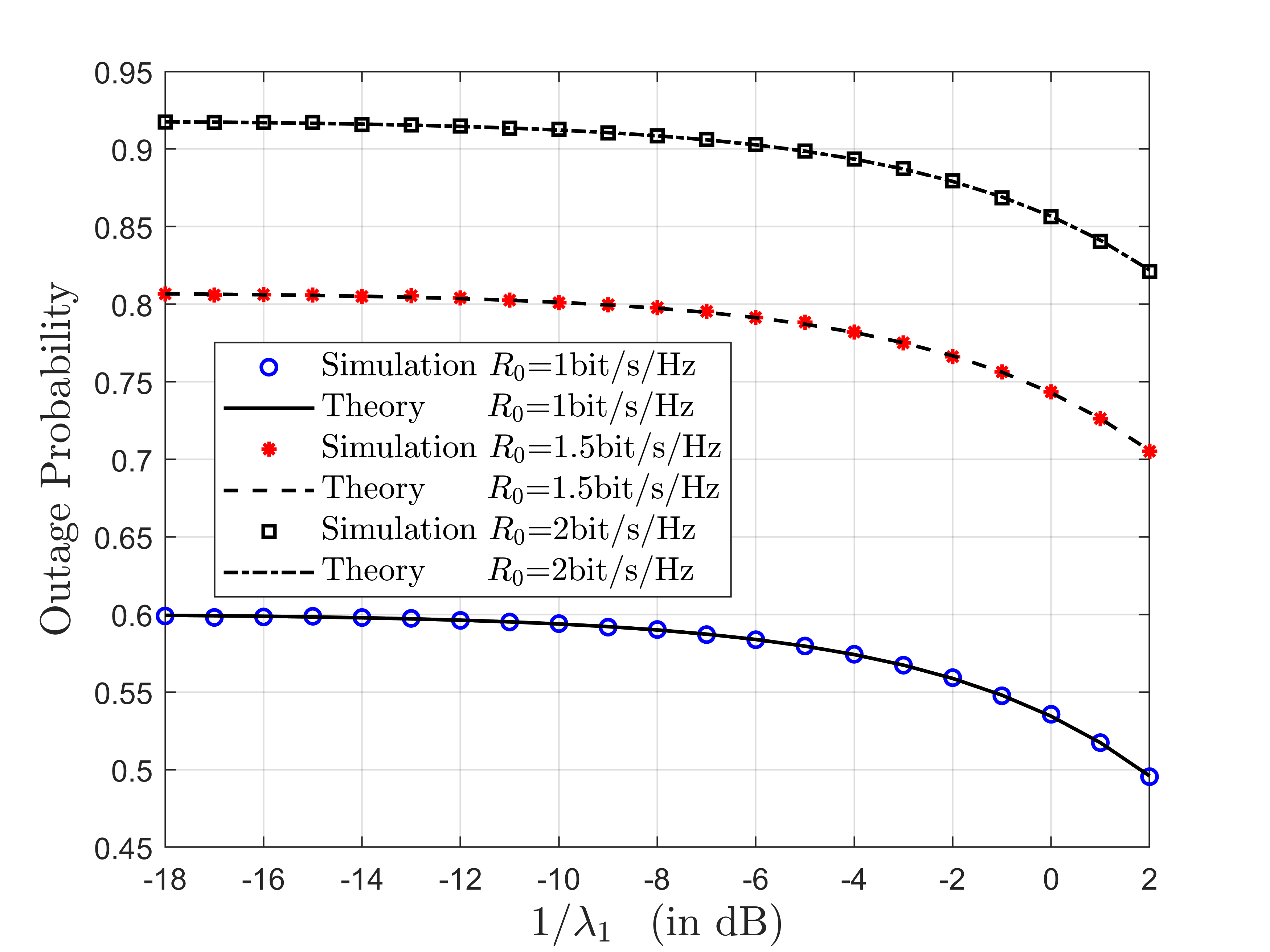}}
\caption{Outage probability of system against relay node $R1$ energy-harvested mean $\frac{1}{\lambda_1}$, with parameters $\frac{1}{\lambda_2}=-5$ dB, $P_S = 10$ dBm, $M_1 = 15$ mJ, $M_2 = 10$ mJ and three different data transmission rates ($R_0$ =1, 1.5 and 2 bit/s/Hz).\label{fig6}}
\end{figure}

\begin{figure}
\centerline{\includegraphics[width=3.5in]{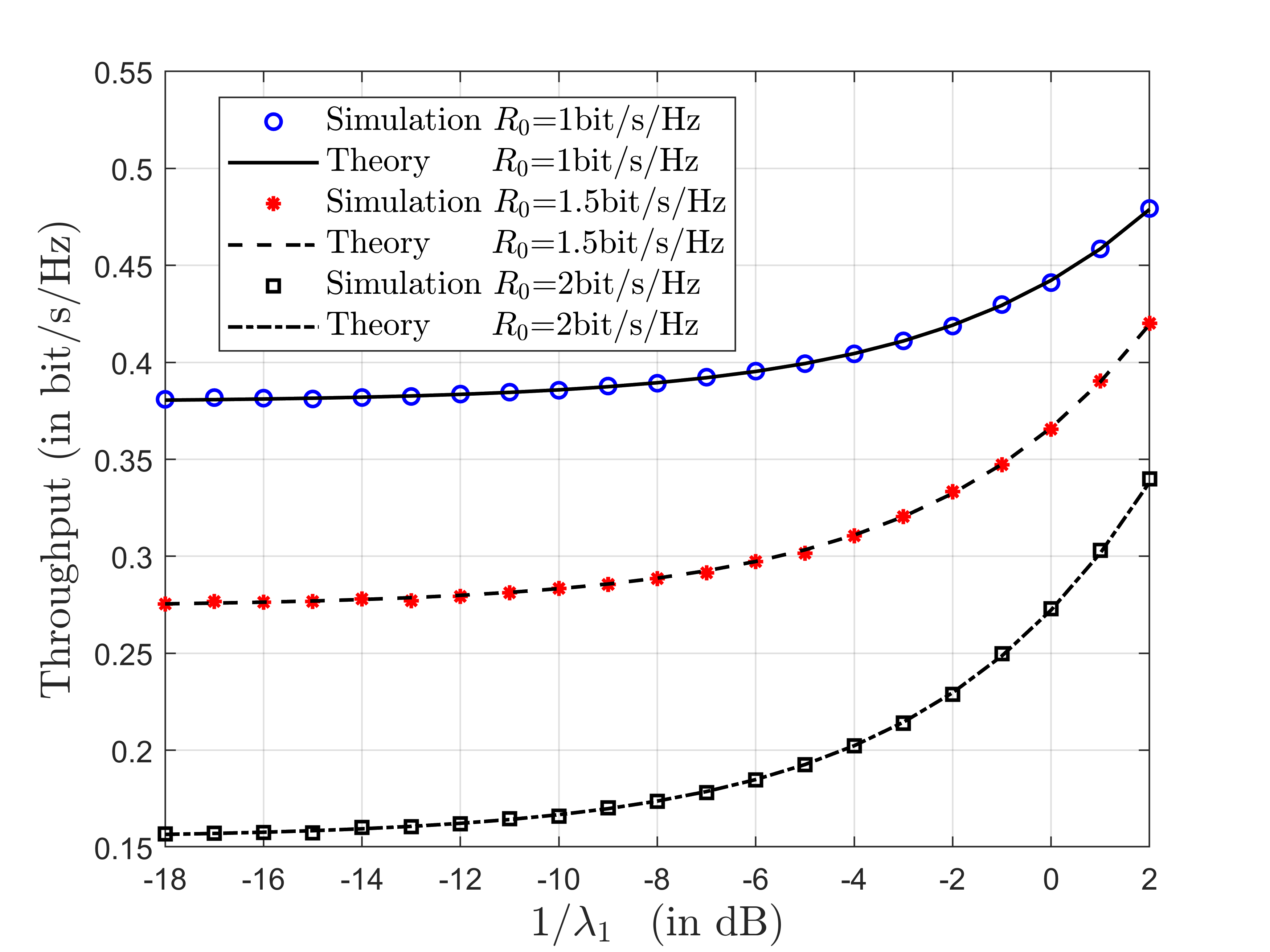}}
\caption{Throughput of system against relay node $R1$ energy-harvested mean $\frac{1}{\lambda_1}$, with parameters $\frac{1}{\lambda_2}=-5$ dB, $P_S = 10$ dBm, $M_1 = 15$ mJ, $M_2 = 10$ mJ, and three different data transmission rates ($R_0$ =1, 1.5 and 2 bit/s/Hz).\label{fig7}}
\end{figure}
Fig. \ref{fig6} and Fig. \ref{fig7} present the variation of system outage probability and system throughput with relay node $R1$ energy-harvested mean $\frac{1}{\lambda_1}$ for three different data transmission rates ($R_0$ =1, 1.5 and 2 bit/s/Hz), respectively. From Fig. \ref{fig6}, it can be seen that the system outage probabilities, which are obtained by simulation and theory, decrease with the increase of $\frac{1}{\lambda_1}$ under the condition that the value of $R_0$ is fixed, on the contrary, decrease with the decrease of the value of $R_0$ when $\frac{1}{\lambda_1}$ is fixed. In contrast, in Fig. \ref{fig7}, the system throughputs, which are obtained by simulation and theory, increase with the increase of $\frac{1}{\lambda_1}$ under the condition that the value of $R_0$ is fixed, and conversely, increase with the decrease of the value of $R_0$ when $\frac{1}{\lambda_1}$ is fixed. This is due to the fact that the increase of $\frac{1}{\lambda_1}$ may lead to the $R1$ having more opportunities to forward the packet, and the decrease of the value of $R_0$ would reduce the threshold $\Gamma_{th}$ which indicates $D$ may successfully receive the packet. Moreover, from Fig. \ref{fig6} and Fig. \ref{fig7}, it can also be seen that the results obtained by theoretical analysis are consistent with the simulation, which effectively verifies the theoretical analysis of system outage probability and system throughput from variations of $\frac{1}{\lambda_1}$ and $R_0$.

\begin{figure}
\centerline{\includegraphics[width=3.5in]{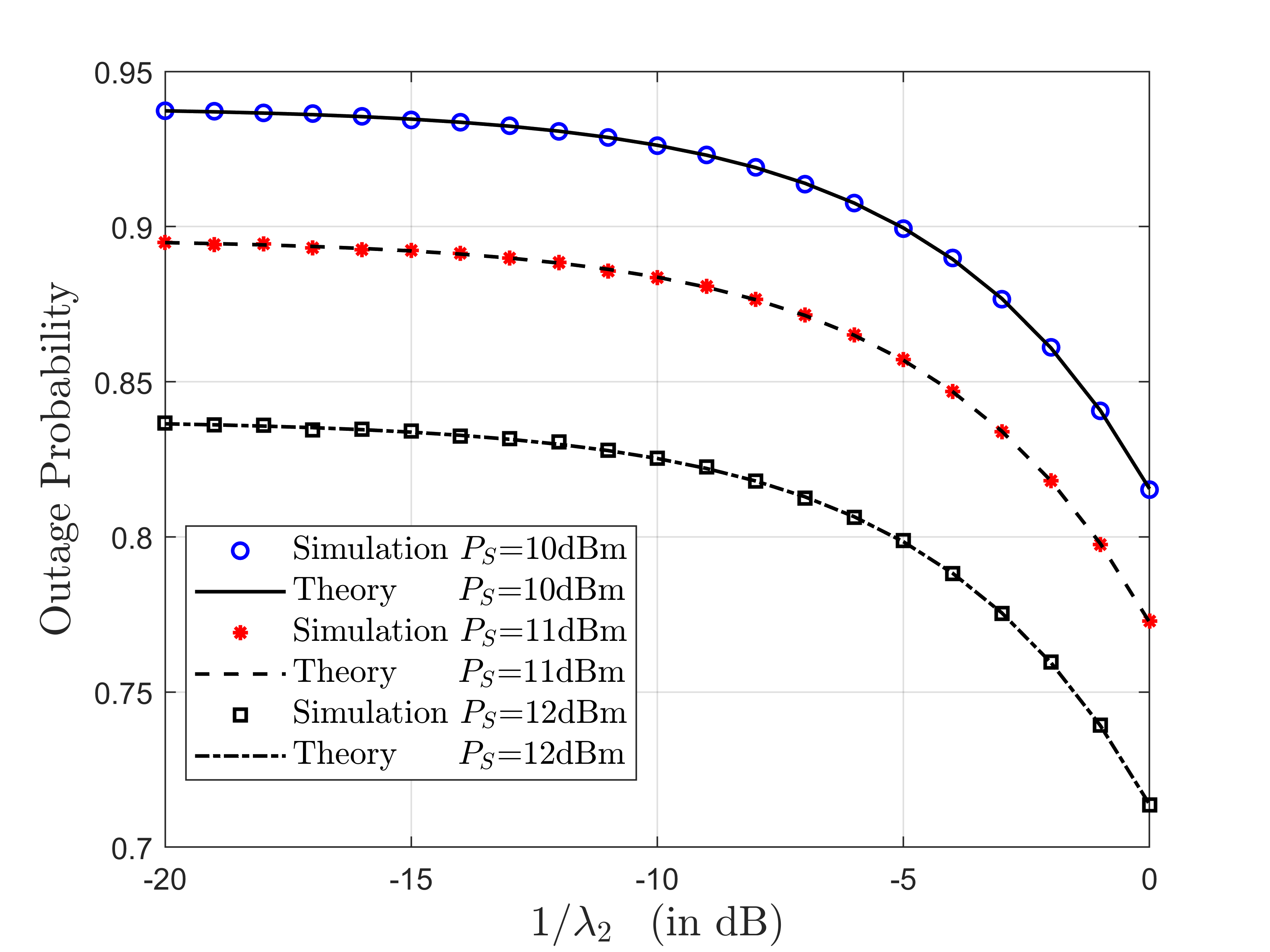}}
\caption{Outage probability of system against relay node $R2$ energy-harvested mean $\frac{1}{\lambda_2}$, with parameters $\frac{1}{\lambda_1}=-10$ dB, $R_0 = 2$ bit/s/Hz, $M_1 = 8$ mJ, $M_2 = 8$ mJ and three different source transmit powers ($P_S$ = 10, 11 and 12 dBm).\label{fig8}}
\end{figure}

\begin{figure}
\centerline{\includegraphics[width=3.5in]{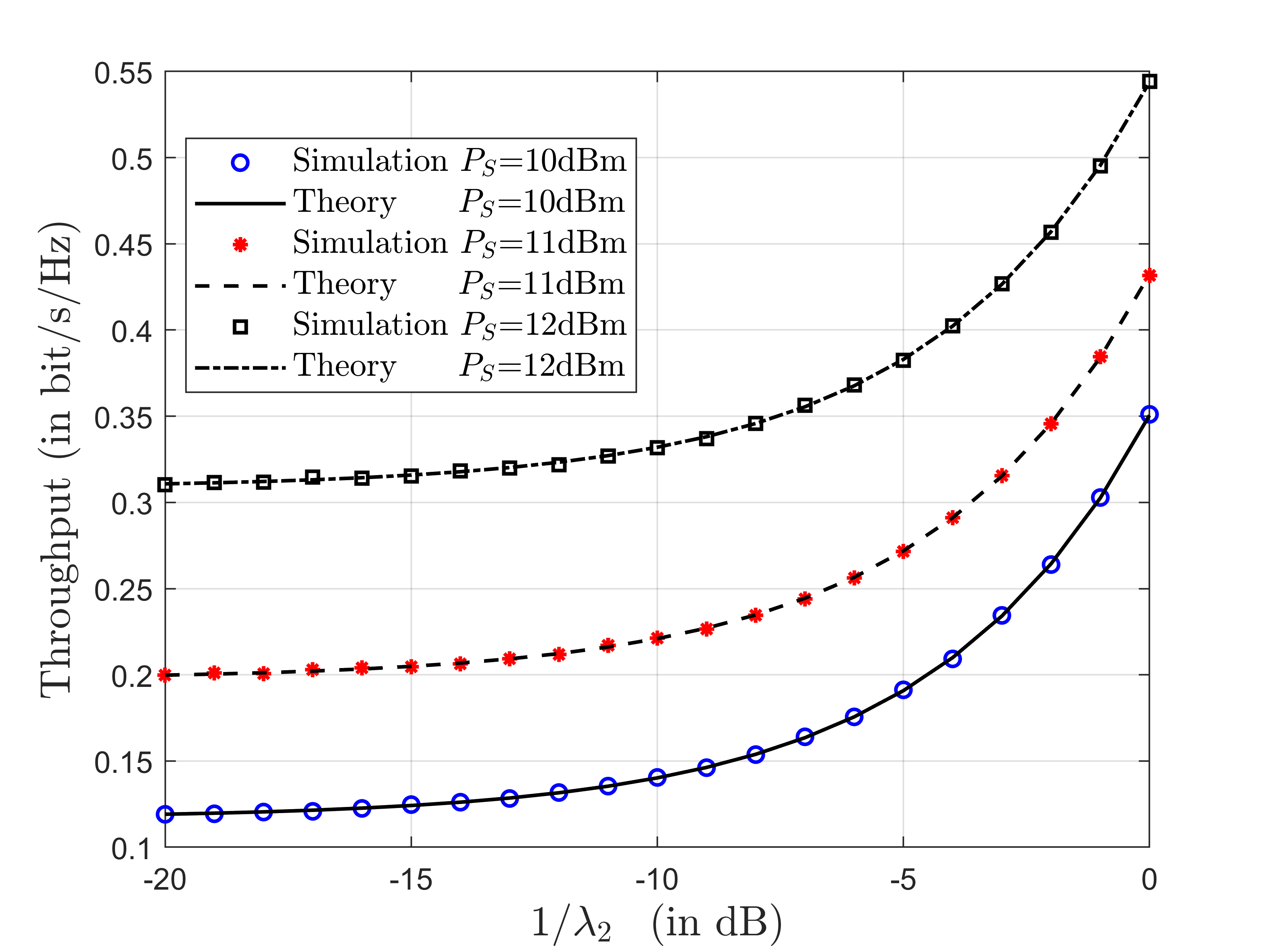}}
\caption{Throughput of system against relay node $R2$ energy-harvested mean $\frac{1}{\lambda_2}$, with parameters $\frac{1}{\lambda_1}=-10$ dB, $R_0 = 2$ bit/s/Hz, $M_1 = 8$ mJ, $M_2 = 8$ mJ and three different source transmit powers ($P_S$ = 10, 11 and 12 dBm).\label{fig9}}
\end{figure}

Fig. \ref{fig8} and \ref{fig9} illustrate the variation of system outage probability and system throughput with relay node $R2$ energy-harvested mean $\frac{1}{\lambda_2}$ for three different source transmit powers ($P_S$ = 10, 11 and 12 dBm), respectively. From Fig. \ref{fig8}, it can be seen that the system outage probabilities, which are obtained by simulation and theory, decrease with the increase of $\frac{1}{\lambda_2}$ under the condition that the value of $P_S$ is fixed, similarly, decrease with the increase of the value of $P_S$ when $\frac{1}{\lambda_2}$ is fixed. However, a comparison of Fig. \ref{fig8} and Fig. \ref{fig9} shows that under the same parameter setting conditions, the changing trend of system throughput in Fig. \ref{fig9} is opposite to that of system outage probability in Fig. \ref{fig8}. This results from that the increase of $\frac{1}{\lambda_2}$ and $P_S$ may increase the probability of $R2$ and $S$ to forward the packet, respectively. So that the probability of $D$ successfully receiving the packet increases. In addition, both Fig. \ref{fig8} and Fig. \ref{fig9} show that the results obtained by theoretical analysis are consistent with the simulation, which effectively verifies the theoretical analysis of system outage probability and system throughput by the variations of $\frac{1}{\lambda_2}$ and $P_S$.

\begin{figure}
\centerline{\includegraphics[width=3.5in]{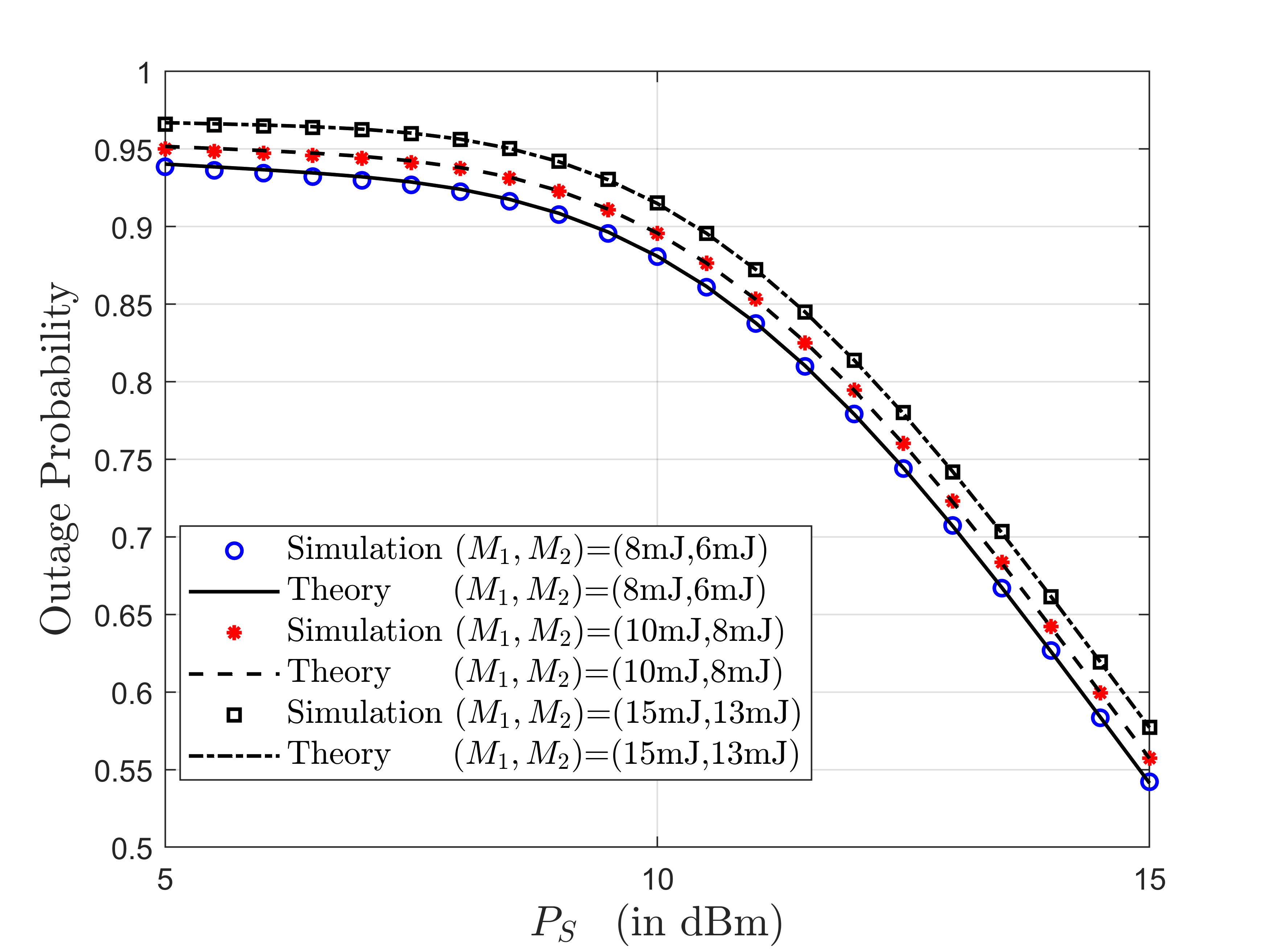}}
\caption{Outage probability of system against source transmit power $P_S$, with parameters $\frac{1}{\lambda_1} = -6$ dB, $\frac{1}{\lambda_2} = -6$ dB, $R_0 = 2$ bit/s/Hz and three different relay energy consumptions \{$(M_1, M_2)$ = (8 mJ, 6 mJ), (10 mJ, 8 mJ) and (15 mJ, 13 mJ)\}.\label{fig10}}
\end{figure}

\begin{figure}
\centerline{\includegraphics[width=3.5in]{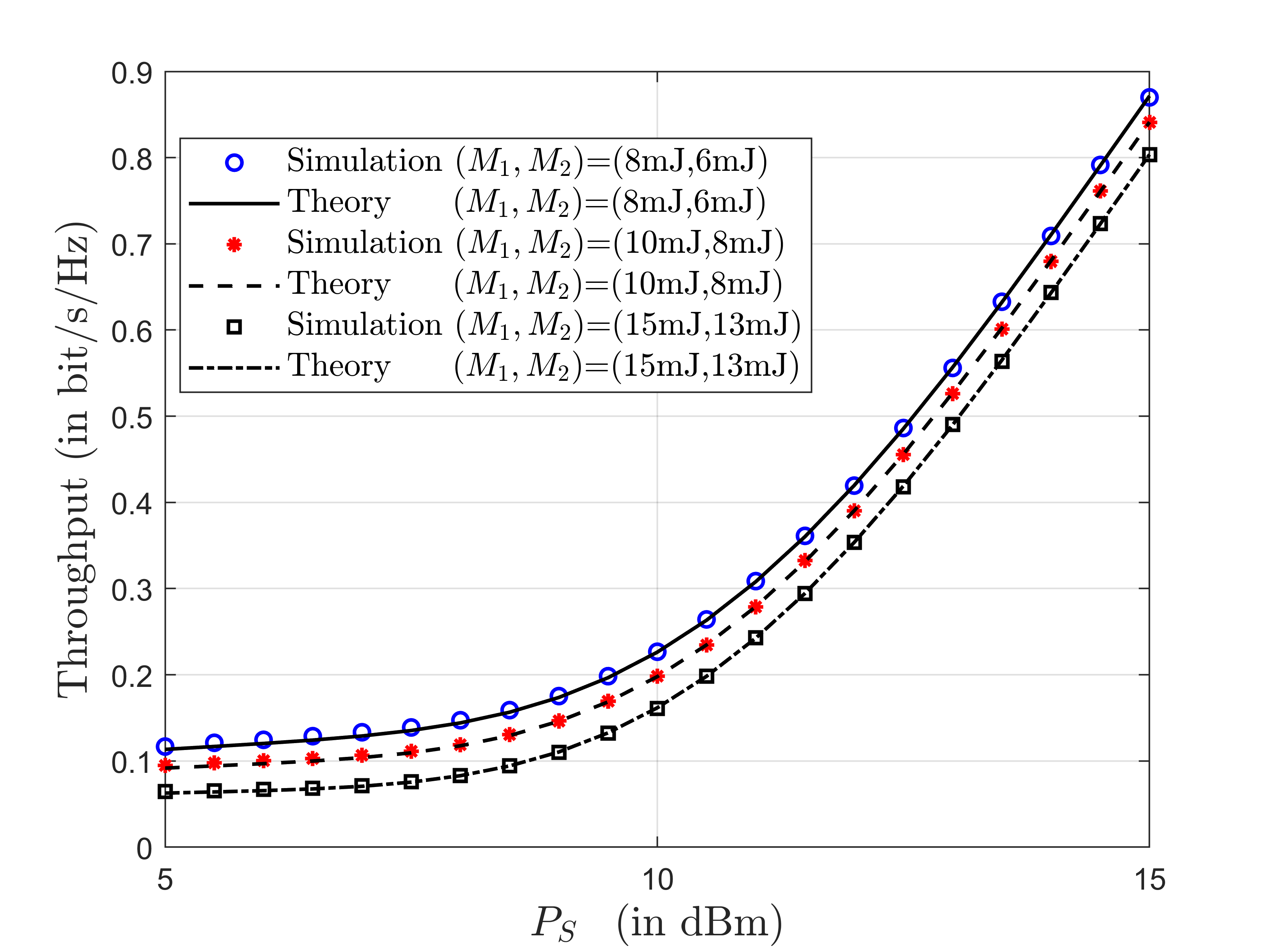}}
\caption{Throughput of system against source transmit power $P_S$, with parameters $\frac{1}{\lambda_1} = -6$ dB, $\frac{1}{\lambda_2} = -6$ dB, $R_0 = 2$ bit/s/Hz and three different relay energy consumptions \{$(M_1, M_2)$ = (8 mJ, 6 mJ), (10 mJ, 8 mJ) and (15 mJ, 13 mJ)\}.\label{fig11}}
\end{figure}

Fig. \ref{fig10} and Fig. \ref{fig11} depict the variation of system outage probability and system throughput with the source transmit power $P_S$ for three different relay energy consumptions \{$(M_1, M_2)$ = (8 mJ, 6 mJ), (10 mJ, 8 mJ) and (15 mJ, 13 mJ)\}, respectively. From Fig. \ref{fig10}, it can be seen that the system outage probabilities, which are obtained by simulation and theory, decrease with the increase of $P_S$ under the condition that the value of $(M_1, M_2)$ is fixed, on the contrary, increase with the increase of the value of $(M_1, M_2)$ when $P_S$ is fixed. Then, comparing Fig. \ref{fig11} with Fig. \ref{fig10}, it is easy to find that under the same parameter setting conditions, the changing trend of system throughput in Fig. \ref{fig11} is opposite to that of system outage probability in Fig. \ref{fig10}. Furthermore, it is also easy to find that when $P_S>10$ dBm, compared with relay energy consumption $(M_1, M_2)$, the source transmits power $P_S$ has a greater impact on system outage probability and system throughput. This is because with the gradual increase of source transmission power $P_S$, the probability of the packet directly transmitted from the source node $S$ to the destination node $D$ increases significantly. Meanwhile, it can be seen from both Fig. \ref{fig10} and Fig. \ref{fig11} that the theoretical results fit into the simulation results, which effectively verifies the theoretical analysis of system outage probability and system throughput from the variations of $P_S$ and $(M_1, M_2)$.

\begin{figure}
\centerline{\includegraphics[width=3.5in]{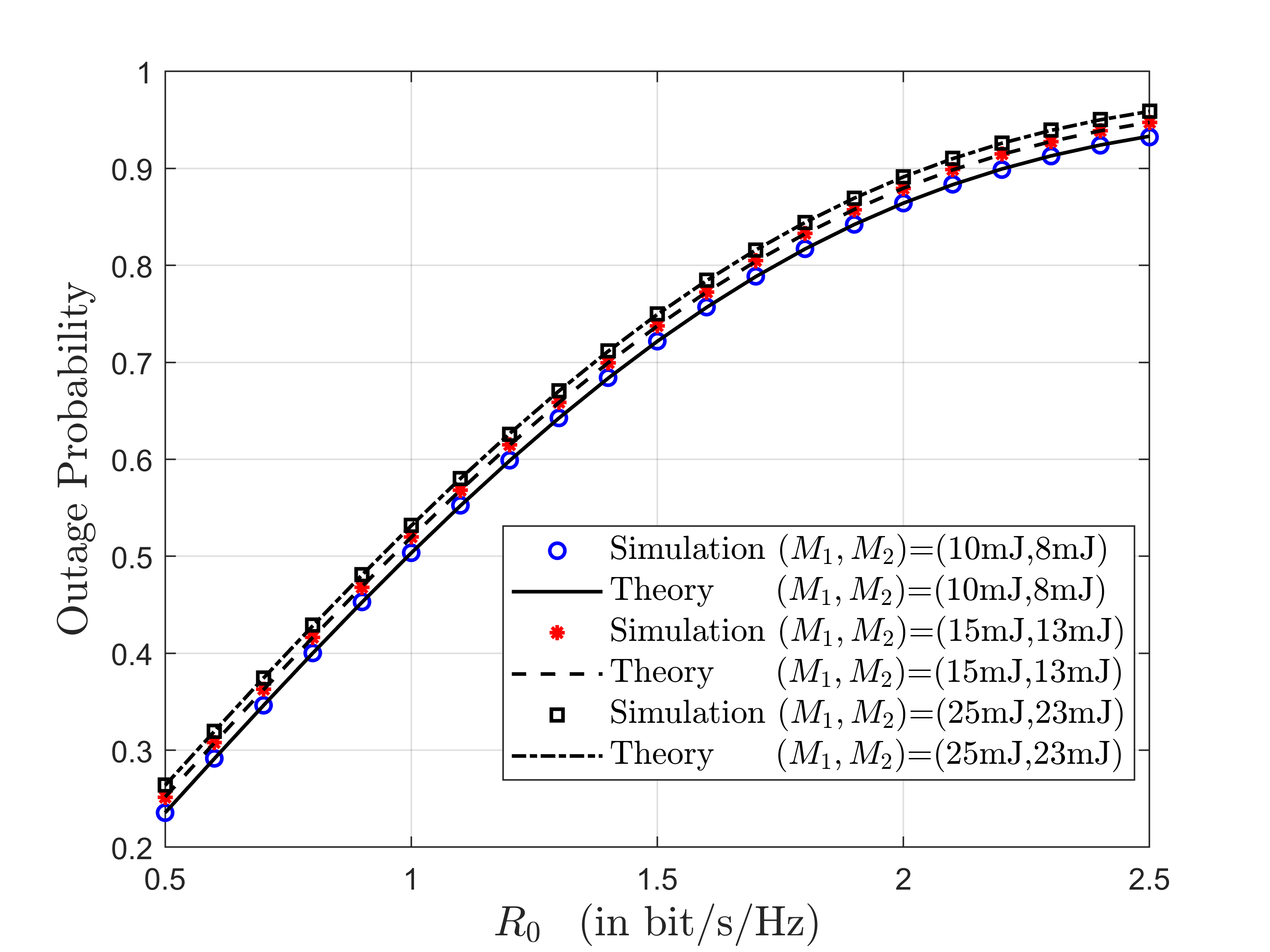}}
\caption{Outage probability of system against data transmission rate $R_0$, with parameters $\frac{1}{\lambda_1} = -7$ dB, $\frac{1}{\lambda_2} = -7$ dB, $P_S = 11$ dBm and three different relay energy consumptions \{$(M_1, M_2)$ = (10 mJ, 8 mJ), (15 mJ, 13 mJ) and (25 mJ, 23 mJ)\}.\label{fig12}}
\end{figure}

\begin{figure}
\centerline{\includegraphics[width=3.5in]{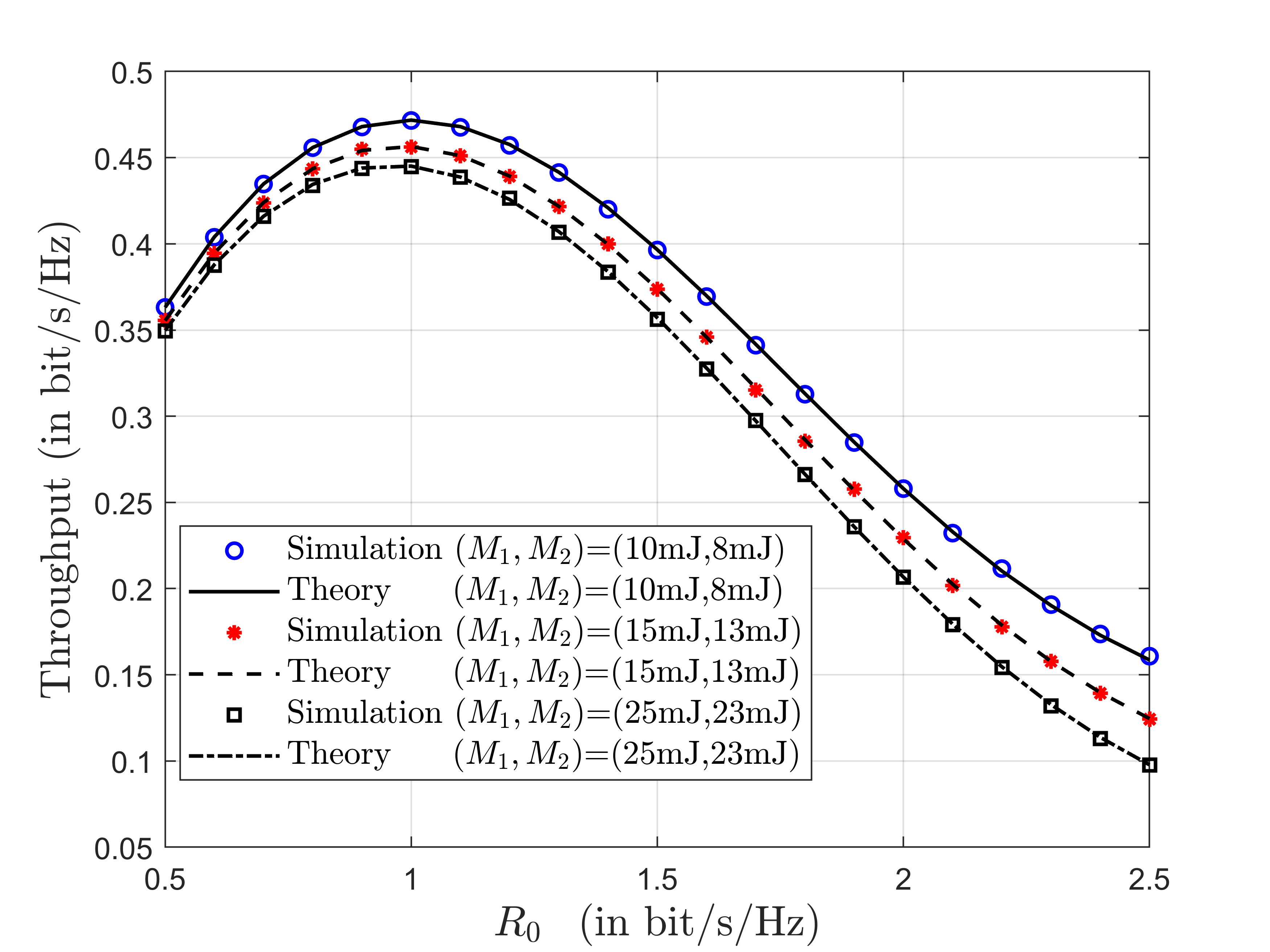}}
\caption{Throughput of system against data transmission rate $R_0$, with parameters $\frac{1}{\lambda_1} = -7$ dB, $\frac{1}{\lambda_2} = -7$ dB, $P_S = 11$ dBm and three different relay energy consumptions \{$(M_1, M_2)$ = (10 mJ, 8 mJ), (15 mJ, 13 mJ) and (25 mJ, 23 mJ)\}.\label{fig13}}
\end{figure}

Fig. \ref{fig12} and Fig. \ref{fig13} present the variation of system outage probability and system throughput with the data transmission rate $R_0$ for three different relay energy consumptions \{$(M_1, M_2)$ = (10 mJ, 8 mJ), (15 mJ, 13 mJ) and (25 mJ, 23 mJ)\}, respectively. From Fig. \ref{fig12}, it can be seen that the system outage probabilities, which are obtained by simulation and theory, increase with the increase of $R_0$ or $(M_1, M_2)$. Especially compared with relay energy consumption $(M_1, M_2)$, data transmission rate $R_0$ has a greater impact on system outage probabilities. This is due to the increase of the value of $R_0$ would significantly increase the threshold $\Gamma_{th}$, which may greatly reduce the probability of node $D$ successfully receiving the packet. However, from Fig. \ref{fig13}, it can be seen that for three different relay energy consumptions \{$(M_1, M_2)$ = (10 mJ, 8 mJ), (15 mJ, 13 mJ) and (25 mJ, 23 mJ)\}, the system throughputs, which are obtained by simulation and theory, increase with the increase of $R_0$ when $R_0 < 1$ bit/s/Hz, and conversely, decrease with the increase of $R_0$ when the $R_0 > 1$ bit/s/Hz. Therefore, it can be concluded that under this system parameter setting, the optimal throughput can be obtained when $R_0$ is about 1 bit/s/Hz. Moreover, both Fig. \ref{fig12} and Fig. \ref{fig13} show that the curves obtained from the theoretical analysis are consistent with the simulation values, which effectively verifies the theoretical analysis of system outage probability and system throughput from variations of $R_0$ and $(M_1, M_2)$.

\begin{figure}
\centerline{\includegraphics[width=3.5in]{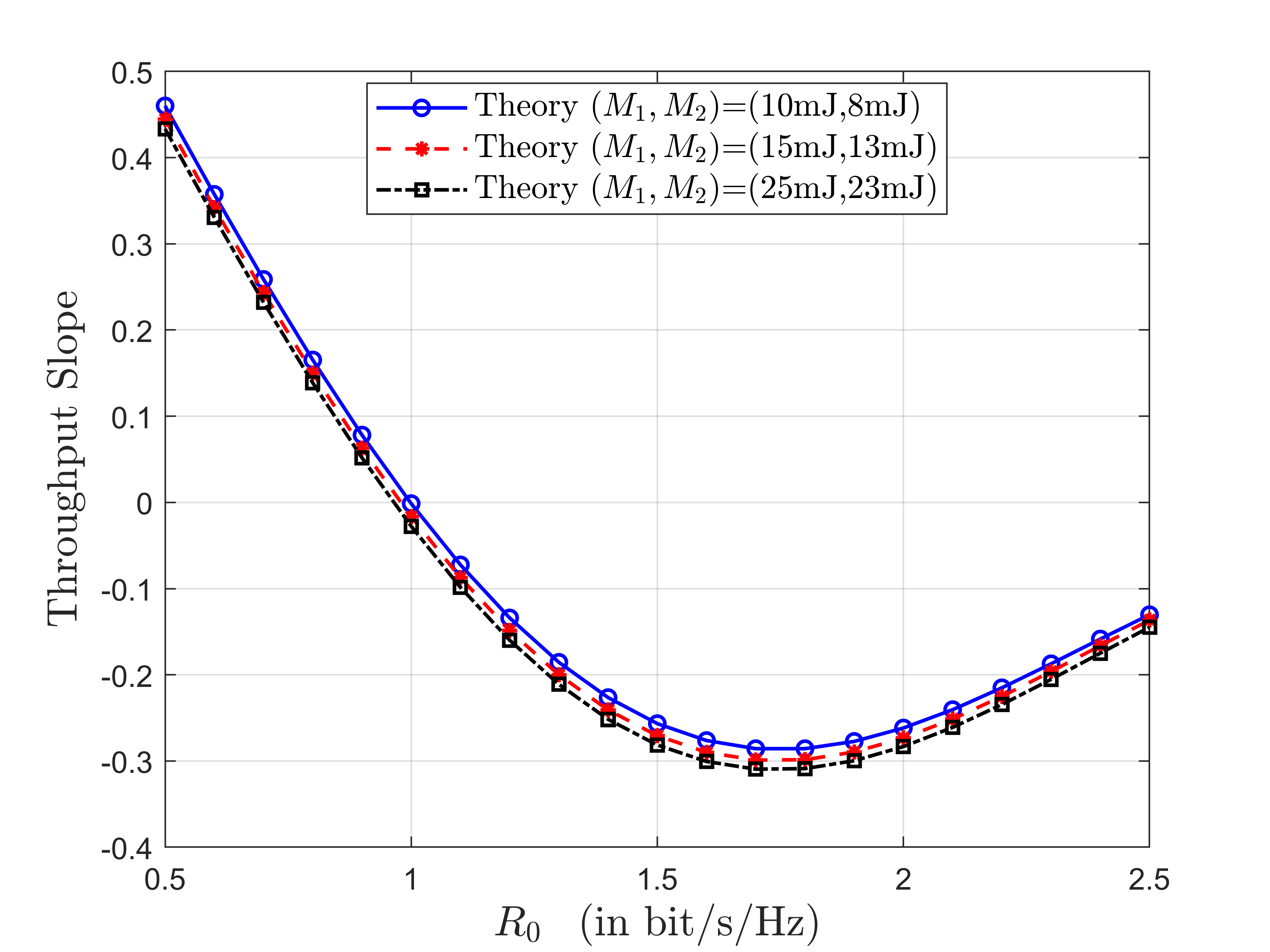}}
\caption{Slope of system theoretical throughput against data transmission rate $R_0$, with parameters $\frac{1}{\lambda_1} = -7$ dB, $\frac{1}{\lambda_2} = -7$ dB, $P_S = 11$ dBm and three different relay energy consumptions \{$(M_1, M_2)$ = (10 mJ, 8 mJ), (15 mJ, 13 mJ) and (25 mJ, 23 mJ)\}.\label{fig14}}
\end{figure}

Fig. \ref{fig14} illustrates the variation of slope of theoretical throughput curves to data transmission rate $R_0$ for three different relay energy consumptions \{$(M_1, M_2)$ = (10 mJ, 8 mJ), (15 mJ, 13 mJ) and (25 mJ, 23 mJ)\}, respectively. From Fig. \ref{fig14}, it can be clearly observed that $R_0 = 1$ bit/s/Hz is approximately the point where the slopes of the three different theoretical throughput curves are zero. Moreover, the slopes corresponding to points on the left of $R_0 = 1$ bit/s/Hz are positive, while the slopes corresponding to points on the right of $R_0 = 1$ bit/s/Hz is negative. That is, when $R_0$ is about 1 bit/s/Hz, the optimal throughput can be obtained, which is consistent with the analysis result in Fig. \ref{fig13}, and also effectively verifies the theoretical analysis of the system throughput.

\section{Conclusion}
This paper proposes the OR aided cooperative communication network with EH. Where the OR protocol is proposed to select the packet transmission path based on the node transmission priority. Additionally, the STM-based algorithm is proposed to find the probability distribution of the CBN set. Using both DCSMC model and the probability distribution, the existence conditions and analytical expressions of the limiting distribution of energy in energy buffers are determined. Then, based on the limiting distribution of energy in energy buffers, the network outage probability and throughput are analyzed, and the corresponding closed-form expressions are given. Furthermore, various simulation findings show that the simulation results are in line with the theoretical analysis results. The relay system containing more than two EH relay nodes is proposed as a future research direction.

\appendices
\section{Proof of Theorem 1}
According to the energy storage process $B_2(i)$ with the infinite-size energy buffer in Eq. (\ref{eq10}), the variable $O_R(i)$ is defined by \cite[Appendix B]{15}
\begin{equation}\label{APPENDIX A1}
            O_R(i) = \begin{cases}
                          1,\quad \mathbb{P}_{22} \\
                          0,\quad \mathbb{P}_{21}
                        \end{cases}
\end{equation}
where $\mathbb{P}_{21}$ and $\mathbb{P}_{22}$ have been given in Eq. (\ref{eq11}) and Eq. (\ref{eq12}), respectively. Implying Eq. (\ref{APPENDIX A1}), the energy storage process in Eq. (\ref{eq10}) can be re-expressed as follows
\begin{equation}\label{APPENDIX A5}
            B_2(i+1) - B_2(i) = X_2(i) - M_2 O_R(i).
\end{equation}
According to the law of large numbers, the average energy harvesting rate can be obtained as follows
\begin{equation}\label{APPENDIX A6}
            \mathbb{E}[X_2(i)] =  \lim_{N \to \infty}\frac{1}{N}\sum_{i=1}^N X_2(i),
\end{equation}
similarly, the average energy consumption rate can be given by
\begin{equation}\label{APPENDIX A7}
    \begin{split}
            \mathbb{E}[M_2 O_R(i)] & = \lim_{N \to \infty}\frac{1}{N}\sum_{i=1}^N M_2 O_R(i) \\
            & \leq M_2\Big\{\lim_{N \to \infty}\frac{1}{N}\sum_{i=1}^N O_R(i)\Big\}.
    \end{split}
\end{equation}
moreover,
\begin{equation}\label{APPENDIX A8}
    \begin{split}
        \mathbb{E}[O_R(i)] & =  \lim_{N \to \infty}\frac{1}{N}\sum_{i=1}^N O_R(i) \\
        & = 1 \times \mathrm{Pr}\{O_R(i) = 1\} + 0 \times \mathrm{Pr}\{O_R(i) = 0\} \\
        & = \mathrm{Pr}\{\textbf{S}(i) = \textbf{s}_3,C_9(i)\} + \mathrm{Pr}\{\textbf{S}(i) = \textbf{s}_4,C_9(i)\} \\
        & \leq \left( p_3+p_4\right)\left(1- e^{-\Omega_{SD}\Gamma_{th}}\right)e^{-\Omega_{R2D}\Gamma_{th}}.
    \end{split}
\end{equation}
From Eq. (\ref{APPENDIX A7}) and Eq. (\ref{APPENDIX A8}), we obtain
\begin{equation}\label{APPENDIX A9}
    \begin{split}
            \mathbb{E}[M_2 O_R(i)] & \leq M_2 \left( p_3+p_4\right)\left(1- e^{-\Omega_{SD}\Gamma_{th}}\right)e^{-\Omega_{R2D}\Gamma_{th}}.
    \end{split}
\end{equation}
If $\psi_1 \leq 1$, from Eq. (\ref{eq13}), we get
\begin{equation}\label{APPENDIX A10}
    \begin{split}
            \mathbb{E}[X_2(i)] \ge  M_2 \left( p_3+p_4\right)\left(1- e^{-\Omega_{SD}\Gamma_{th}}\right)e^{-\Omega_{R2D}\Gamma_{th}}.
    \end{split}
\end{equation}
Therefore, if $\psi_1 \leq 1$, from Eq. (\ref{APPENDIX A6}), Eq. (\ref{APPENDIX A7}), Eq. (\ref{APPENDIX A9}) and Eq. (\ref{APPENDIX A10}), we obtain
\begin{equation}\label{APPENDIX A11}
            \lim_{N \to \infty}\frac{1}{N}\sum_{i=1}^N X_2(i) \ge \lim_{N \to \infty}\frac{1}{N}\sum_{i=1}^N M_2 O_R(i).
\end{equation}
According to Eq. (\ref{APPENDIX A5}) and Eq. (\ref{APPENDIX A11}), we have
\begin{equation}\label{APPENDIX A12}
            \lim_{N \to \infty}\frac{1}{N}\sum_{i=1}^N B_2(i+1) \ge \lim_{N \to \infty}\frac{1}{N}\sum_{i=1}^N B_2(i).
\end{equation}
Clearly, if the inequality condition holds in Eq. (\ref{APPENDIX A12}), the comparison of between $B_2(i+1)$ and $B_2(i)$ shows that the energy accumulates in the buffer over time slot, i.e., $\begin{matrix} \lim_{i \to \infty}\mathbb{E}[B_2(i)] = \infty\end{matrix}$. Therefore, the stationary distribution of $B_2(i)$ does not exist, and after a finite number of time slots, $B_2(i) > M_2$ almost always hold \cite[Appendix A]{34}. In addition, if the equality condition holds in Eq. (\ref{APPENDIX A11}), according to Eq. (\ref{APPENDIX A11}) and Eq. (\ref{APPENDIX A7}), we get
\begin{equation}\label{APPENDIX A13}
            \lim_{N \to \infty}\frac{1}{N}\sum_{i=1}^N X_2(i) = \lim_{N \to \infty}\frac{1}{N}\sum_{i=1}^N M_2 O_R(i).
\end{equation}
Eq. (\ref{APPENDIX A13}) indicates that in the energy buffer with the DCSMC model, the average energy harvesting rate equals the average energy consumption rate, which is unstable \cite{29}. Therefore, the buffer may almost always provide $M_2$ amount of energy.

\section{Proof of Theorem 2}
According to the total probability theorem, $\mathbb{P}_{21}$ and $\mathbb{P}_{22}$ together constitutes a complete set of events accompanying event $B_2(i+1)$. Therefore, the cumulative distribution function (CDF) of $B_2(i+1)$ in storage process in Eq. (\ref{eq10}) may be evaluated as follows \cite[Appendix C]{15}
\begin{equation}\label{APPENDIX B1}
    \begin{split}
         & \mathrm{Pr}\{B_2(i+1) \leq x\} \\
         & \quad = \mathrm{Pr}\{B_2(i)+X_2(i) \leq x, \textbf{S}(i) = \textbf{s}_1\} \\
         & \quad + \mathrm{Pr}\{B_2(i)+X_2(i) \leq x, \textbf{S}(i) = \textbf{s}_2\} \\
         & \quad + \mathrm{Pr}\{B_2(i)+X_2(i) \leq x, \textbf{S}(i) = \textbf{s}_3, C_1(i) \} \\
         & \quad + \mathrm{Pr}\{B_2(i)+X_2(i) \leq x, \textbf{S}(i) = \textbf{s}_3,\overline{C_1}(i), B_2(i)< M_2 \} \\
         & \quad + \mathrm{Pr}\{B_2(i)+X_2(i) \leq x, \textbf{S}(i) = \textbf{s}_3,\overline{C_1}(i), B_2(i)\ge M_2, \\
         & \quad \gamma_{R2D}(i)< \Gamma_{th} \} \\
         & \quad + \mathrm{Pr}\{B_2(i)+X_2(i) \leq x, \textbf{S}(i) = \textbf{s}_4, C_1(i) \} \\
         & \quad + \mathrm{Pr}\{B_2(i)+X_2(i) \leq x, \textbf{S}(i) = \textbf{s}_4,\overline{C_1}(i), B_2(i)< M_2 \} \\
         & \quad + \mathrm{Pr}\{B_2(i)+X_2(i) \leq x, \textbf{S}(i) = \textbf{s}_4,\overline{C_1}(i), B_2(i)\ge M_2, \\
         & \quad \gamma_{R2D}(i)< \Gamma_{th} \} \\
         & \quad + \mathrm{Pr}\{B_2(i)+X_2(i)-M_2 \leq x, \textbf{S}(i) = \textbf{s}_3,C_9(i) \} \\
         & \quad + \mathrm{Pr}\{B_2(i)+X_2(i)-M_2 \leq x, \textbf{S}(i) = \textbf{s}_4,C_9(i) \}.
    \end{split}
\end{equation}
Let $B_2(i)=\mu_2$, $X_2(i)=X_2$, $\textbf{S}(i) = \textbf{S}$, $C_1(i) = C_1$, $\overline{C_1}(i) = \overline{C_1}$, $C_9(i) = C_9$, and $\gamma_{R2D}(i)=\gamma_{R2D}$. Eq. (\ref{APPENDIX B1}) may be presented by
\begin{equation}\label{APPENDIX B2}
    \begin{split}
             & \mathrm{Pr}\{B_2(i+1) \leq x\} \\
             & \quad = \mathrm{Pr}\{\mu_2+X_2 \leq x, \textbf{S} = \textbf{s}_1\} \\
             & \quad + \mathrm{Pr}\{\mu_2+X_2 \leq x, \textbf{S} = \textbf{s}_2\} \\
             & \quad + \mathrm{Pr}\{\mu_2+X_2 \leq x, \textbf{S} = \textbf{s}_3, C_1 \} \\
             & \quad + \mathrm{Pr}\{\mu_2+X_2 \leq x, \textbf{S} = \textbf{s}_3,\overline{C_1}, B_2(i)< M_2 \} \\
             & \quad + \mathrm{Pr}\{\mu_2+X_2 \leq x, \textbf{S} = \textbf{s}_3,\overline{C_1}, B_2\ge M_2,\gamma_{R2D}< \Gamma_{th} \} \\
             & \quad + \mathrm{Pr}\{\mu_2+X_2 \leq x, \textbf{S} = \textbf{s}_4, C_1 \} \\
             & \quad + \mathrm{Pr}\{\mu_2+X_2 \leq x, \textbf{S} = \textbf{s}_4,\overline{C_1}, \mu_2< M_2 \} \\
             & \quad + \mathrm{Pr}\{\mu_2+X_2 \leq x, \textbf{S} = \textbf{s}_4,\overline{C_1}, \mu_2\ge M_2,\gamma_{R2D}< \Gamma_{th} \} \\
             & \quad + \mathrm{Pr}\{\mu_2+X_2-M_2 \leq x, \textbf{S} = \textbf{s}_3,C_9 \} \\
             & \quad + \mathrm{Pr}\{\mu_2+X_2-M_2 \leq x, \textbf{S} = \textbf{s}_4,C_9 \} .
    \end{split}
\end{equation}
Denote $G^{i+1}_2(x)$ as the CDF of $B_2(i+1)$, namely, $\mathrm{Pr}\{B_2(i+1) \leq x\} = G^{i+1}_2(x)$. When $i \to \infty$, if the buffer reaches its steady state, it follows that, $\mathrm{Pr}\{B_2(i+1) \leq x\} = G^{i+1}_2(x) = G^{i}_2(x) = G_2(x)$. In this state, Eq. (\ref{APPENDIX B2}) may be written as follows
\begin{equation}\label{APPENDIX B3}
    \begin{split}
         G_2(x) & = (p_1+p_2) \int_{\mu_2=0}^{x} F_{X_2}(x-\mu_2)g_2(\mu_2)\, d\mu_2 \\
         & + (p_3+p_4) \int_{r_1=\Gamma_{th}}^{\infty}\int_{\mu_2=0}^{x} F_{X_2}(x-\mu_2)g_2(\mu_2) \\
         & \times f_{\gamma_{SD}}(r_1)\, d\mu_2\, dr_1 \\
         & + (p_3+p_4) \int_{r_1=0}^{\Gamma_{th}}\int_{\mu_2=0}^{min(x,M_2)} F_{X_2}(x-\mu_2)g_2(\mu_2) \\
         & \times f_{\gamma_{SD}}(r_1)\, d\mu_2\, dr_1 \\
         & + (p_3+p_4) \int_{r_1=0}^{\Gamma_{th}}\int_{r_2=0}^{\Gamma_{th}}\int_{\mu_2=M_2}^{x} F_{X_2}(x-\mu_2)g_2(\mu_2) \\               & \times f_{\gamma_{SD}}(r_1)f_{\gamma_{R2D}}(r_2)\, d\mu_2\, dr_1\, dr_2 \\
         & +(p_3+p_4)\int_{r_1=0}^{\Gamma_{th}}\int_{r_2=\Gamma_{th}}^{\infty}\int_{\mu_2=M_2}^{x+M_2} F_{X_2}(x+M_2-\mu_2) \\
         & \times g_2(\mu_2)f_{\gamma_{SD}}(r_1)f_{\gamma_{R2D}}(r_2)\, d\mu_2\, dr_1\, dr_2,
    \end{split}
\end{equation}
where $F_{X_2}(x)=1-e^{-\lambda_2 x}$ is the CDF of $X_2$. $g_2(x)$ is the derivative of $G_2(x)$. In addition, $f_{\gamma_{SD}}(r_1) = \Omega_{SD}e^{-\Omega_{SD}r_1}$ and $f_{\gamma_{R2D}}(r_2) = \Omega_{R2D}e^{-\Omega_{R2D}r_2}$ are the PDFs of $\gamma_{SD}$ and $\gamma_{R2D}$, respectively. Substituting $f_{\gamma_{SD}}(r_1)$ and $f_{\gamma_{R2D}}(r_2)$ in Eq. (\ref{APPENDIX B3}), and simplifying Eq. (\ref{APPENDIX B3}), we arrive at
\begin{equation}\label{APPENDIX B4}
            G_2(x) = \begin{cases}
                          G_{21}(x) , & 0 \leq x < M_2 \\
                          G_{22}(x), & x \ge M_2
                        \end{cases}
\end{equation}
where,
\begin{equation}\label{APPENDIX B5}
    \begin{split}
         G_{21}(x) & = \int_{\mu_2=0}^{x} F_{X_2}(x-\mu_2)g_2(\mu_2)\, d\mu_2 \\
         & +(p_3+p_4)\left(1-e^{-\Omega_{SD}\Gamma_{th}}\right)e^{-\Omega_{R2D}\Gamma_{th}} \\
         & \times \int_{\mu_2=M_2}^{x+M_2} F_{X_2}(x+M_2-\mu_2)g_2(\mu_2)\, d\mu_2,
    \end{split}
\end{equation}
\begin{equation}\label{APPENDIX B6}
    \begin{split}
         G_{22}(x) & = \int_{\mu_2=0}^{M_2} F_{X_2}(x-\mu_2)g_2(\mu_2)\, d\mu_2 \\
         & + \left[1-\left(p_3+p_4\right)\left(1-e^{-\Omega_{SD}\Gamma_{th}}\right)e^{-\Omega_{R2D}\Gamma_{th}}\right] \\                   & \times \int_{\mu_2=M_2}^{x} F_{X_2}(x-\mu_2)g_2(\mu_2)\, d\mu_2 \\
         & +\left[\left(p_3+p_4\right)\left(1-e^{-\Omega_{SD}\Gamma_{th}}\right)e^{-\Omega_{R2D}\Gamma_{th}}\right]\\
         & \times \int_{\mu_2=M_2}^{x+M_2} F_{X_2}(x+M_2-\mu_2) g_2(\mu_2)\, d\mu_2.
    \end{split}
\end{equation}
According to Eq. (\ref{APPENDIX B4}), the PDF $g_2(x)$ may be defined as
\begin{equation}\label{APPENDIX B7}
            g_2(x) = \begin{cases}
                          g_{21}(x) , & 0 \leq x < M_2 \\
                          g_{22}(x), & x \ge M_2.
                        \end{cases}
\end{equation}
After substituting Eq. (\ref{APPENDIX B7}) into Eq. (\ref{APPENDIX B6}), the derivatives about $x$ on both sides of Eq. (\ref{APPENDIX B6}) may be obtained
\begin{equation}\label{APPENDIX B8}
    \begin{split}
             g_{22}(x) & = \int_{\mu_2=0}^{M_2} f_{X_2}(x-\mu_2)g_{21}(\mu_2)\, d\mu_2 \\
             & \quad + a_2 \int_{\mu_2=M_2}^{x} f_{X_2}(x-\mu_2)g_{22}(\mu_2)\, d\mu_2  \\
             & \quad + b_2 \int_{\mu_2=M_2}^{x+M_2} f_{X_2}(x+M_2-\mu_2)g_{22}(\mu_2)\, d\mu_2,
    \end{split}
\end{equation}
where,
\begin{equation}\label{APPENDIX B9}
            a_2=\left[1-\left(p_3+p_4\right)\left(1-e^{-\Omega_{SD}\Gamma_{th}}\right)e^{-\Omega_{R2D}\Gamma_{th}}\right],
\end{equation}
\begin{equation}\label{APPENDIX B10}
            b_2=\left[\left(p_3+p_4\right)\left(1-e^{-\Omega_{SD}\Gamma_{th}}\right)e^{-\Omega_{R2D}\Gamma_{th}}\right].
\end{equation}
Clearly, $a_2+b_2 = 1$. Furthermore, there is evidence that storage process in Eq. (\ref{eq10}) possesses a unique stationary distribution \cite[Appendix B]{34}, when $M_2>\mathbb{E}[X_2(i)]=1/\lambda_2$. In other words, $g_2(x)$ has unique solution. Furthermore, $g_{22}(x)$ be assumed to have an exponential-type solution expressed by $g_{22}(x) = k_2 e^{Q_2 x}$ \cite[Appendix C]{34}. Substituting $g_{22}(x) = k_2 e^{Q_2 x}$ and $f_{X_2}(x) = \lambda_2 e^{-\lambda_2 x}$ into Eq. (\ref{APPENDIX B8}), we obtain
\begin{equation}\label{APPENDIX B11}
    \begin{split}
             k_2 e^{Q_2 x} & = \int_{\mu_2=0}^{M_2} \lambda_2 e^{-\lambda_2 (x-\mu_2)} g_{21}(\mu_2)\, d\mu_2 \\
             & \quad + a_2 \int_{\mu_2=M_2}^{x} \lambda_2 e^{-\lambda_2 (x-\mu_2)} k_2 e^{Q_2 \mu_2}\, d\mu_2 \\
             & \quad + b_2 \int_{\mu_2=M_2}^{x+M_2} \lambda_2 e^{-\lambda_2 (x+M_2-\mu_2)} k_2 e^{Q_2 \mu_2}\, d\mu_2.
    \end{split}
\end{equation}
Simplifying Eq. (\ref{APPENDIX B11}), we have
\begin{equation}\label{APPENDIX B12}
    \begin{split}
             k_2 e^{Q_2 x} & = \lambda_2 e^{-\lambda_2 x}\int_{\mu_2=0}^{M_2} e^{\lambda_2 \mu_2} g_{21}(\mu_2)\, d\mu_2 \\
             & \quad - \dfrac{\lambda_2 k_2 \left[ b_2 + a_2 e^{\lambda_2 M_2}\right]}{\lambda_2 + Q_2} e^{\left(Q_2 M_2-\lambda_2 x\right)}\\
             & \quad + \dfrac{b_2 \lambda_2 e^{Q_2 M_2} + a_2 \lambda_2}{\lambda_2 + Q_2} k_2 e^{Q_2 x}.
    \end{split}
\end{equation}
For the correctness of postulating $g_{22}(x) = k_2 e^{Q_2 x}$, both sides of Eq. (\ref{APPENDIX B12}) must be equal. Hence ,we get
\begin{subequations}
    \begin{numcases}{}
          \dfrac{b_2 \lambda_2 e^{Q_2 M_2} + a_2 \lambda_2}{\lambda_2 + Q_2} = 1, \label{APPENDIX B13a}\\ 
          \dfrac{k_2 e^{Q_2 M_2} \left[ b_2 + a_2 e^{\lambda_2 M_2}\right]}{\lambda_2 + Q_2} = \int_{\mu_2=0}^{M_2} e^{\lambda_2 \mu_2} g_{21}(\mu_2)\, d\mu_2. \label{APPENDIX B13b}
    \end{numcases}
\end{subequations}
It can be seen from Eq. (\ref{APPENDIX B13a}) that $Q_{2_0} = 0$ is one of the solutions of $Q_2$ in Eq. (\ref{APPENDIX B13a}), but this solution does not meet the condition that $g_{22}(x)$ is a finite distribution. Additionally, the other solution $Q_{2_1}$ of $Q_2$ in Eq. (\ref{APPENDIX B13a}) can be obtained by simplifying Eq. (\ref{APPENDIX B13a}) as
\begin{equation}\label{APPENDIX B14}
             b_2 \lambda_2 e^{Q_2 M_2} = \lambda_2 - a_2 \lambda_2 + Q_2= b_2 \lambda_2 + Q_2.
\end{equation}
Using Lambert W function, the solution $Q_{2_1}$ of $Q_2$ can be obtained as follows
\begin{equation}\label{APPENDIX B15}
             Q_{2_1} = \frac{-W\left(-b_2\lambda_2 M_2 e^{-b_2\lambda_2 M_2}\right)}{M_2}-b_2\lambda_2.
\end{equation}
According to the property of Lambert W function, when $b_2\lambda_2 M_2 \leq 1$, $W\left(-b_2\lambda_2 M_2 e^{-b_2\lambda_2 M_2}\right)= -b_2\lambda_2 M_2$ so that $Q_{2_1}=Q_{2_0}=0$. On the contrary, when $b_2\lambda_2 M_2 >1$, $W\left(-b_2\lambda_2 M_2 e^{-b_2\lambda_2 M_2}\right)> -b_2\lambda_2 M_2$ so that $Q_{2_1}<0$, ensuring the finite distribution of $g_{22}(x)$. Thus, for the stationary distribution of $g_{22}(x)$, we obtain
\begin{equation}\label{APPENDIX B16}
             Q_2 = \frac{-W\left(-b_2\lambda_2 M_2 e^{-b_2\lambda_2 M_2}\right)}{M_2}-b_2\lambda_2,\quad b_2\lambda_2 M_2>1.
\end{equation}

Similarly, when $0\leq x <M_2$, substituting Eq. (\ref{APPENDIX B7}) into Eq. (\ref{APPENDIX B5}), the derivatives about $x$ on both sides of Eq. (\ref{APPENDIX B5}) can be obtained
\begin{equation}\label{APPENDIX B17}
    \begin{split}
             g_{21}(x) & = b_2 \int_{\mu_2=M_2}^{x+M_2} f_{X_2}(x+M_2-\mu_2)g_{22}(\mu_2)\, d\mu_2\\
             & \quad + \int_{\mu_2=0}^{x} f_{X_2}(x-\mu_2)g_{21}(\mu_2)\, d\mu_2.
    \end{split}
\end{equation}
Substituting $g_{22}(x)=k_2 e^{Q_2 x}$ and $f_{X_2}(x)=\lambda_2 e^{-\lambda_2 x}$ into Eq. (\ref{APPENDIX B17}), we get
\begin{equation}\label{APPENDIX B18}
    \begin{split}
             g_{21}(x) & = \lambda_2 \int_{\mu_2=0}^{x} e^{-\lambda_2 \left(x-\mu_2 \right)}g_{21}(\mu_2)\, d\mu_2 \\
             & \quad + \dfrac{b_2 k_2 \lambda_2 e^{Q_2 M_2}}{\lambda_2 + Q_2} \left(e^{Q_2 x}-e^{-\lambda_2 x} \right).
    \end{split}
\end{equation}
Let $r_2(x)=\frac{b_2 k_2 \lambda_2 e^{Q_2 M_2}}{\lambda_2 + Q_2} \left(e^{Q_2 x}-e^{-\lambda_2 x} \right)$, and the integral equation in Eq. (\ref{APPENDIX B18}) can be rewritten as follows
\begin{equation}\label{APPENDIX B19}
    \begin{split}
             g_{21}(x) = \lambda_2 \int_{\mu_2=0}^{x} e^{-\lambda_2 \left(x-\mu_2 \right)}g_{21}(\mu_2)\, d\mu_2 + r_2(x).
    \end{split}
\end{equation}
Clearly, Eq. (\ref{APPENDIX B19}) is a Volterra integral equation of the second kind, whose solution is given by \cite{6}, \cite{15} and \cite[eq. 2.2.1]{30}
\begin{equation}\label{APPENDIX B20}
    \begin{split}
             g_{21}(x) = r_2(x) + \lambda_2 \int_{t=0}^{x} r_2(t)\, dt.
    \end{split}
\end{equation}
Substituting $r_2(x)$ into Eq. (\ref{APPENDIX B20}), we obtain
\begin{equation}\label{APPENDIX B21}
    \begin{split}
             g_{21}(x) & = \dfrac{b_2 k_2 \lambda_2 e^{Q_2 M_2}}{\lambda_2 + Q_2} \left(e^{Q_2 x}-e^{-\lambda_2 x} \right) \\
             & \quad + \lambda_2 \int_{t=0}^{x} \frac{b_2 k_2 \lambda_2 e^{Q_2 M_2}}{\lambda_2 + Q_2} \left(e^{Q_2 t}-e^{-\lambda_2 t} \right)\, dt \\
             & = \dfrac{b_2 k_2 \lambda_2 e^{Q_2 M_2}\left(e^{Q_2 x}-1 \right)}{Q_2}.
    \end{split}
\end{equation}
According to the unit area condition on $g_2(x)$, we have
\begin{equation}\label{APPENDIX B22}
             \int_{x=0}^{\infty} g_2(x)\, dx = \int_{x=0}^{M_2} g_{21}(x)\, dx +\int_{x=M_2}^{\infty} g_{22}(x)\, dx = 1.
\end{equation}
Substituting $g_{21}(x) = \frac{b_2 k_2 \lambda_2 e^{Q_2 M_2}\left(e^{Q_2 x}-1 \right)}{Q_2}$ and $g_{22}(x) = k_2 e^{Q_2 x}$ into Eq. (\ref{APPENDIX B22}), we get
\begin{equation}\label{APPENDIX B23}
             \dfrac{b_2 k_2 \lambda_2 e^{Q_2 M_2}}{Q_2} \int_{x=0}^{M_2} \left(e^{Q_2 x}-1 \right)\, dx +  k_2 \int_{x=M_2}^{\infty} e^{Q_2 x}\, dx =1.
\end{equation}
Simplifying Eq. (\ref{APPENDIX B23}), we have
\begin{equation}\label{APPENDIX B24}
             \dfrac{b_2 k_2 \lambda e^{Q_2 M_2}}{Q_2} \left[\dfrac{e^{Q_2 M_2}-1}{Q_2}-M_2 \right] - \dfrac{k_2 e^{Q_2 M_2}}{Q_2} =1.
\end{equation}
Substituting Eq. (\ref{APPENDIX B14}) into Eq. (\ref{APPENDIX B24}), then simplifying Eq. (\ref{APPENDIX B24}), the value of $k_2$ can be obtained as follows
\begin{equation}\label{APPENDIX B25}
              k_2 = \dfrac{-Q_2}{M_2 \left(b_2 \lambda_2 + Q_2\right)}.
\end{equation}
Substituting Eq. (\ref{APPENDIX B14}) and Eq. (\ref{APPENDIX B25}) into Eq. (\ref{APPENDIX B21}), we arrive at
\begin{equation}\label{APPENDIX B26}
               g_{21}(x) =  \dfrac{1-e^{Q_2 x}}{M_2}.
\end{equation}
Substituting Eq. (\ref{APPENDIX B26}) into the right side of Eq. (\ref{APPENDIX B13b}), we obtain
\begin{equation}\label{APPENDIX B27}
    \begin{split}
             \int_{\mu_2=0}^{M_2} e^{\lambda_2 \mu_2} g_{21}(\mu_2)\, d\mu_2 & = \int_{\mu_2=0}^{M_2} \dfrac{\left(1-e^{Q_2 x}\right)e^{\lambda_2 \mu_2}}{M_2}\, d\mu_2 \\
             & = \dfrac{1-e^{\left(\lambda_2+Q_2\right)M_2}}{\left(\lambda_2+Q_2\right)M_2} - \dfrac{1-e^{\lambda_2 M_2}}{\lambda_2 M_2}.
    \end{split}
\end{equation}
The equation in Eq. (\ref{APPENDIX B14}) leads us to conclude $\lambda_2 M_2 = \frac{\left(\lambda_2+Q_2\right)M_2}{b_2 e^{Q_2 M_2} + a_2}$. Substituting this conclusion in Eq. (\ref{APPENDIX B27}), we have
\begin{equation}\label{APPENDIX B28}
    \begin{split}
             \int_{\mu_2=0}^{M_2} e^{\lambda_2 \mu_2} g_{21}(\mu_2)\, d\mu_2 = \dfrac{\left(1-e^{Q_2 M_2}\right)\left(b_2+a_2 e^{\lambda_2 M_2}\right)}{\left(\lambda_2+Q_2\right)M_2}.
    \end{split}
\end{equation}
Similarly, the conclude $1-e^{Q_2 M_2} = \frac{-Q_2}{b_2\lambda_2}$ may be obtained from Eq. (\ref{APPENDIX B14}). Substituting this conclusion in Eq. (\ref{APPENDIX B28}), we arrive at
\begin{equation}\label{APPENDIX B29}
    \begin{split}
             \int_{\mu_2=0}^{M_2} e^{\lambda_2 \mu_2} g_{21}(\mu_2)\, d\mu_2 & = \dfrac{-Q_2 \left(b_2+a_2 e^{\lambda_2 M_2}\right)}{\left(\lambda_2+Q_2\right)M_2 b_2\lambda_2} \\
             & = \dfrac{-Q_2 e^{Q_2 M_2} \left(b_2+a_2 e^{\lambda_2 M_2}\right)}{\left(\lambda_2+Q_2\right)M_2 b_2\lambda_2 e^{Q_2 M_2}} \\
             & = \dfrac{-Q_2 e^{Q_2 M_2} \left(b_2+a_2 e^{\lambda_2 M_2}\right)}{M_2 \left(b_2\lambda_2 + Q_2\right) \left(\lambda_2+Q_2\right)} \\
             & = \dfrac{k_2 e^{Q_2 M_2} \left(b_2+a_2 e^{\lambda_2 M_2}\right)}{\lambda_2+Q_2}.
    \end{split}
\end{equation}
Now, it can be shown that $g_{21}(x)$ satisfies the condition in Eq. (\ref{APPENDIX B13b}). Therefore, there is no doubt that the unique solution $g_{21}(x)$ in Eq. (\ref{APPENDIX B26}) for Eq. (\ref{APPENDIX B17}) and the unique solution $g_{22}(x) = k_2 e^{Q_2 x}$ for Eq. (\ref{APPENDIX B8}) are obtained.

\section{Proof of Theorem 3}
According to the total probability theorem, $\mathbb{P}_{11}$ and $\mathbb{P}_{12}$,  the CDF of $B_1(i+1)$ in storage process in Eq. (\ref{eq15}) may be evaluated as follows
\begin{equation}\label{APPENDIX C 00}
    \begin{split}
         & \mathrm{Pr}\{B_1(i+1) \leq x\} \\
         & \quad = \mathrm{Pr}\{B_1(i)+X_1(i) \leq x, \textbf{S}(i) = \textbf{s}_1\} \\
         & \quad + \mathrm{Pr}\{B_1(i)+X_1(i) \leq x, \textbf{S}(i) = \textbf{s}_3\} \\
         & \quad + \mathrm{Pr}\{\Theta_1\} + \mathrm{Pr}\{\Theta_2\} \\
         & \quad + \mathrm{Pr}\{B_1(i)+X_1(i)-M_1 \leq x,\textbf{S}(i)=\textbf{s}_2, C_5(i)\}\\
         & \quad + \mathrm{Pr}\{B_1(i)+X_1(i)-M_1 \leq x,\textbf{S}(i)=\textbf{s}_2, C_8(i)\}\\
         & \quad + \mathrm{Pr}\{B_1(i)+X_1(i)-M_1 \leq x,\textbf{S}(i)=\textbf{s}_4, C_{10}(i)\}\\
         & \quad + \mathrm{Pr}\{B_1(i)+X_1(i)-M_1 \leq x,\textbf{S}(i)=\textbf{s}_4, C_{11}(i)\},
    \end{split}
\end{equation}
where, $\mathrm{Pr}\{\Theta_1\}$ is presented as follows
\begin{equation}\label{APPENDIX C 01}
    \begin{split}
         & \mathrm{Pr}\{\Theta_1\} \\
         & \quad =  \mathrm{Pr}\{B_1(i)+X_1(i) \leq x, \textbf{S}(i) = \textbf{s}_2, C_1(i) \} \\
         & \quad + \mathrm{Pr}\{B_1(i)+X_1(i) \leq x, \textbf{S}(i) = \textbf{s}_2,\overline{C_1}(i), B_1(i)< M_1 \} \\
         & \quad + \mathrm{Pr}\{B_1(i)+X_1(i) \leq x, \textbf{S}(i) = \textbf{s}_2,\overline{C_1}(i), B_1(i)\ge M_1, \\
         & \quad \gamma_{R1D}(i)< \Gamma_{th},\gamma_{SR2}(i)\ge \Gamma_{th} \} \\
         & \quad + \mathrm{Pr}\{B_1(i)+X_1(i) \leq x, \textbf{S}(i) = \textbf{s}_2,\overline{C_1}(i), B_1(i)\ge M_1, \\
         & \quad \gamma_{R1D}(i)< \Gamma_{th},\gamma_{SR2}(i)\ge \Gamma_{th},\gamma_{R1R2}(i)< \Gamma_{th}, \}, \\
    \end{split}
\end{equation}
and $\mathrm{Pr}\{\Theta_2\}$ is given by
\begin{equation}\label{APPENDIX C 0}
    \begin{split}
         & \mathrm{Pr}\{\Theta_2\} \\
         & \quad =  \mathrm{Pr}\{B_1(i)+X_1(i) \leq x, \textbf{S}(i) = \textbf{s}_4, C_1(i) \} \\
         & \quad + \mathrm{Pr}\{B_1(i)+X_1(i) \leq x, \textbf{S}(i) = \textbf{s}_4, C_9(i) \} \\
         & \quad + \mathrm{Pr}\{B_1(i)+X_1(i) \leq x, \textbf{S}(i) = \textbf{s}_4,\overline{C_1}(i), B_2(i)\ge M_2, \\
         & \quad \gamma_{R2D}(i)< \Gamma_{th},B_1(i)\ge M_1,\gamma_{R1D}(i)< \Gamma_{th} \} \\
         & \quad + \mathrm{Pr}\{B_1(i)+X_1(i) \leq x, \textbf{S}(i) = \textbf{s}_4,\overline{C_1}(i), B_2(i)\ge M_2, \\
         & \quad \gamma_{R2D}(i)< \Gamma_{th},B_1(i)< M_1 \} \\
         & \quad + \mathrm{Pr}\{B_1(i)+X_1(i) \leq x, \textbf{S}(i) = \textbf{s}_4,\overline{C_1}(i), B_2(i)< M_2, \\
         & \quad B_1(i)\ge M_1,\gamma_{R1D}(i)< \Gamma_{th} \} \\
         & \quad + \mathrm{Pr}\{B_1(i)+X_1(i) \leq x, \textbf{S}(i) = \textbf{s}_4,\overline{C_1}(i), B_2(i)< M_2, \\
         & \quad B_1(i)< M_1 \}.
    \end{split}
\end{equation}
Since $\gamma_{SD}$, $\gamma_{SR2}$, $\gamma_{R1D}$, $\gamma_{R1R2}$ and $\gamma_{R2D}$ obey the exponential distribution with parameters $\Omega_{SD}$, $\Omega_{SR2}$, $\Omega_{R1D}$, $\Omega_{R1R2}$ and $\Omega_{R2D}$, respectively. In addition, the probabilities of $\mathrm{Pr}\{B_2(i)\ge M_R\}$ has been given in Eq. (\ref{eq141}). When $i \to \infty$, Eq. (\ref{APPENDIX C 00}) is written as follows
\begin{equation}\label{APPENDIX C 2}
    \begin{split}
         G_1(x) & = a_{11} \int_{\mu_1=0}^{x} F_{X_1}(x-\mu_1)g_1(\mu_1)\, d\mu_1 \\
         & \quad + a_{12} \int_{\mu_1=0}^{min(x,M_1)} F_{X_1}(x-\mu_1)g_1(\mu_1)\, d\mu_1 \\
         & \quad + a_{13} \int_{\mu_1=M_1}^{x} F_{X_1}(x-\mu_1)g_1(\mu_1)\, d\mu_1 \\
         & \quad + b_1 \int_{\mu_1=M_1}^{x+M_1} F_{X_1}(x+M_1-\mu_1)g_1(\mu_1)\, d\mu_1,
    \end{split}
\end{equation}
where $g_1(x)$ is the derivative of $G_1(x)$, $F_{X_1}(x)=1-e^{-\lambda_1 x}$ is the CDF of $X_1$. In addition, $a_{11}$ is denoted by
\begin{equation}\label{APPENDIX C 21}
    \begin{split}
         a_{11} & = p_1 + p_3 + \left(p_2 + p_4\right)e^{-\Omega_{SD}\Gamma_{th}} \\
         & \quad + \frac{p_4 \left(1-e^{-\Omega_{SD}\Gamma_{th}}\right)e^{-\Omega_{R2D}\Gamma_{th}}}{M_2 b_2 \lambda_2},
    \end{split}
\end{equation}
$a_{12}$ is given by
\begin{equation}\label{APPENDIX C 22}
    \begin{split}
         a_{12} & = \bigg[p_2 + p_4 \left(1-\frac{1}{M_2 b_2 \lambda_2}\right)\bigg]\left(1-e^{-\Omega_{SD}\Gamma_{th}}\right)\\
         & \quad + \frac{p_4 \left(1-e^{-\Omega_{SD}\Gamma_{th}}\right)\left(1- e^{-\Omega_{R2D}\Gamma_{th}}\right)}{M_2 b_2 \lambda_2},
    \end{split}
\end{equation}
$a_{13}$ is written by
\begin{equation}\label{APPENDIX C 23}
    \begin{split}
         a_{13} = & \left(1-e^{-\Omega_{SD}\Gamma_{th}}\right)\left(1-e^{-\Omega_{R1D}\Gamma_{th}}\right)\\
         & \times \bigg[p_2-p_2\left(1-e^{-\Omega_{SR2}\Gamma_{th}}\right)e^{-\Omega_{R1R2}\Gamma_{th}} \\
         & + p_4\left(1-\frac{e^{-\Omega_{R2D}\Gamma_{th}}}{M_2 b_2 \lambda_2}\right) \bigg],
    \end{split}
\end{equation}
and $b_1$ is denoted by
\begin{equation}\label{APPENDIX C 24}
    \begin{split}
         b_1 &=  \left(1-e^{-\Omega_{SD}\Gamma_{th}}\right)\bigg[p_2e^{-\Omega_{R1D}\Gamma_{th}}\\
         & + p_2\left(1-e^{-\Omega_{R1D}\Gamma_{th}}\right)\left(1-e^{-\Omega_{SR2}\Gamma_{th}}\right)e^{-\Omega_{R1R2}\Gamma_{th}} \\
         & + p_4e^{-\Omega_{R1D}\Gamma_{th}}\left(1-\frac{e^{-\Omega_{R2D}\Gamma_{th}}}{M_2 b_2 \lambda_2}\right) \bigg].
    \end{split}
\end{equation}
Furthermore, the $g_1(x)$ is defined as follows
\begin{equation}\label{APPENDIX C 6}
            g_1(x) = \begin{cases}
                          g_{11}(x), & 0 \leq x < M_1 \\
                          g_{12}(x), & x \ge M_1.
                        \end{cases}
\end{equation}
When $x \ge M_1$, substituting Eq. (\ref{APPENDIX C 6}) into Eq. (\ref{APPENDIX C 2}), the derivatives about $x$ on both sides of Eq. (\ref{APPENDIX C 2}) can be obtained as follows
\begin{equation}\label{APPENDIX C 7}
    \begin{split}
             g_{12}(x) & = \left(a_{11}+a_{12}\right)\int_{\mu_1=0}^{M_1} f_{X_1}(x-\mu_1)g_{11}(\mu_1)\, d\mu_1 \\
             & + \left(a_{11}+a_{13}\right) \int_{\mu_1=M_1}^{x} f_{X_1}(x-\mu_1)g_{12}(\mu_1)\, d\mu_1 \\
             & + b_1 \int_{\mu_1=M_1}^{x+M_1} f_{X_1}(x+M_1-\mu_1)g_{12}(\mu_1)\, d\mu_1.
    \end{split}
\end{equation}
From Eq. (\ref{APPENDIX C 21}) to Eq. (\ref{APPENDIX C 24}), it is easy to know
\begin{equation}\label{APPENDIX C 71}
    \begin{split}
             a_{11}+a_{12}  = a_{11}+a_{13}+b_1 = p_1+p_2+p_3+p_4=1.
    \end{split}
\end{equation}
Like $g_{22}(x)$ in appendix B, let $g_{12}(x) = k_1 e^{Q_1 x}$. Moreover, substituting $g_{12}(x) = k_1 e^{Q_1 x}$ and $f_{X_1}(x) = \lambda_1 e^{-\lambda_1 x}$ and Eq. (\ref{APPENDIX C 71}) into Eq. (\ref{APPENDIX C 7}), we get
\begin{equation}\label{APPENDIX C 10}
    \begin{split}
         k_1 e^{Q_1 x} &= \int_{\mu_1=0}^{M_1} \lambda_1 e^{-\lambda_1 (x-\mu_1)} g_{11}(\mu_1)\, d\mu_1 \\
         & + \left(1-b_1\right) \int_{\mu_1=M_1}^{x} \lambda_1 e^{-\lambda_1 (x-\mu_1)} k_1 e^{Q_1 \mu_1}\, d\mu_1 \\
         & + b_1 \int_{\mu_1=M_1}^{x+M_1} \lambda_1 e^{-\lambda_1 (x+M_1-\mu_1)} k_1 e^{Q_1 \mu_1}\, d\mu_1.
    \end{split}
\end{equation}
Let $a_{11}+a_{13} = a_1$ and simplify Eq. (\ref{APPENDIX C 10}), we obtain
\begin{equation}\label{APPENDIX C 11}
    \begin{split}
             k_1 e^{Q_1 x} & = \lambda_1 e^{-\lambda_1 x}\int_{\mu_1=0}^{M_1} e^{\lambda_1 \mu_1} g_{11}(\mu_1)\, d\mu_1 \\
             & \quad - \dfrac{k_1 \lambda_1 e^{\left(Q_1 M_1-\lambda_1 x \right)} \left(a_1 e^{\lambda_1 M_1}+b_1 \right)}{\lambda_1 + Q_1} \\
             & \quad + \dfrac{\lambda_1 \left( a_1 + b_1 e^{Q_1 M_1} \right)k_1 e^{Q_1 x}}{\lambda_1 + Q_1}.
    \end{split}
\end{equation}
For Eq. (\ref{APPENDIX C 11}) to hold, the following conditions need to be satisfied
\begin{subequations}
    \begin{numcases}{}
          \dfrac{\lambda_1 \left( a_1 + b_1 e^{Q_1 M_1} \right)}{\lambda_1 + Q_1} = 1, \label{APPENDIX C 12a}\\
          \dfrac{k_1 e^{Q_1 M_1} \left(a_1 e^{\lambda_1 M_1}+b_1 \right)}{\lambda_1 + Q_1} = \int_{\mu_1=0}^{M_1} e^{\lambda_1 \mu_1} g_{11}(\mu_1)\, d\mu_1. \label{APPENDIX C 12b}
    \end{numcases}
\end{subequations}
Similarly, the desirable solution $Q_1$ of Eq. (\ref{APPENDIX C 12a}) for the finite distribution of $g_{12}(x)$ may be obtained by simplifying Eq. (\ref{APPENDIX C 12a}) as follows
\begin{equation}\label{APPENDIX C 13}
             b_1 \lambda_1 e^{Q_1 M_1} = \lambda_1 - a_1 \lambda_1 + Q_1= b_1 \lambda_1 + Q_1.
\end{equation}
Using Lambert W function, the solution $Q_1$ of Eq. (\ref{APPENDIX C 13}) may be obtained as follows
\begin{equation}\label{APPENDIX C 14}
             Q_1 = \frac{-W\left(-b_1\lambda_1 M_1 e^{-b_1\lambda_1 M_1}\right)}{M_1}-b_1\lambda_1, \quad b_1\lambda_1 M_1>1,
\end{equation}
where, since $b_1\lambda_1 M_1 >1$, $W\left(-b_1\lambda_1 M_1 e^{-b_1\lambda_1 M_1}\right)> -b_1\lambda_1 M_1$. Furthermore, $Q_1<0$, which ensures the finite distribution of $g_{12}(x)$.

When $0\leq x <M_1$, substituting Eq. (\ref{APPENDIX C 6}) into Eq. (\ref{APPENDIX C 2}), the derivatives about $x$ on both sides of Eq. (\ref{APPENDIX C 2}) can be obtained as follows
\begin{equation}\label{APPENDIX C 16}
    \begin{split}
             g_{11}(x) & = b_1 \int_{\mu_1=M_1}^{x+M_1} f_{X_1}(x+M_1-\mu_1)g_{12}(\mu_1)\, d\mu_1\\
             & \quad + \int_{\mu_1=0}^{x} f_{X_1}(x-\mu_1)g_{11}(\mu_1)\, d\mu_1.
    \end{split}
\end{equation}
Substituting $g_{12}(x)=k_1 e^{Q_1 x}$ and $f_{X_1}(x)=\lambda_1 e^{-\lambda_1 x}$ into Eq. (\ref{APPENDIX C 16}), we get
\begin{equation}\label{APPENDIX C 17}
    \begin{split}
             g_{11}(x) & = \lambda_1 \int_{\mu_1=0}^{x} e^{-\lambda_1 \left(x-\mu_1 \right)}g_{11}(\mu_1)\, d\mu_1 \\
             & \quad + \dfrac{b_1 k_1 \lambda_1 e^{Q_1 M_1}}{\lambda_1 + Q_1} \left(e^{Q_1 x}-e^{-\lambda_1 x} \right).
    \end{split}
\end{equation}
Similar to $ g_{21}(x)$ in appendix B, the solution of $g_{11}(x)$ may be given as follows
\begin{equation}\label{APPENDIX C 20}
    \begin{split}
             g_{11}(x) & = \dfrac{b_1 k_1 \lambda_1 e^{Q_1 M_1}}{\lambda_1 + Q_1} \left(e^{Q_1 x}-e^{-\lambda_1 x} \right) \\
             & \quad + \lambda_1 \int_{t=0}^{x} \dfrac{b_1 k_1 \lambda_1 e^{Q_1 M_1}}{\lambda_1 + Q_1} \left(e^{Q_1 t}-e^{-\lambda_1 t} \right)\, dt \\
             & = \dfrac{b_1 k_1 \lambda_1 e^{Q_1 M_1}\left(e^{Q_1 x}-1 \right)}{Q_1}.
    \end{split}
\end{equation}
Because of the unit area condition on $g_1(x)$, we have
\begin{equation}\label{APPENDIX C 21}
             \int_{x=0}^{\infty} g_1(x)\, dx = \int_{x=0}^{M_1} g_{11}(x)\, dx +\int_{x=M_1}^{\infty} g_{12}(x)\, dx = 1.
\end{equation}
Substituting $g_{11}(x) = \frac{b_1 k_1 \lambda_1 e^{Q_1 M_1}\left(e^{Q_1 x}-1 \right)}{Q_1}$ and $g_{12}(x) = k_1 e^{Q_1 x}$ into Eq. (\ref{APPENDIX C 21}), we arrive at
\begin{equation}\label{APPENDIX C 24}
              k_1 = \dfrac{-Q_1}{M_1 \left(b_1 \lambda_1 + Q_1\right)}.
\end{equation}
Furthermore, according to Eq. (\ref{APPENDIX C 24}), we have
\begin{equation}\label{APPENDIX C 25}
               g_{11}(x) =  \dfrac{1-e^{Q_1 x}}{M_1}.
\end{equation}
Similar to the validation of Eq. (\ref{APPENDIX B13b}), according to Eq. (\ref{APPENDIX C 13}), Eq. (\ref{APPENDIX C 24}) and Eq. (\ref{APPENDIX C 25}), the Eq. (\ref{APPENDIX C 12b}) may be validated as follows
\begin{equation}\label{APPENDIX C 28}
    \begin{split}
             \int_{\mu_1=0}^{M_1} e^{\lambda_1 \mu_1} g_{11}(\mu_1)\, d\mu_1 & =\dfrac{\left(1-e^{Q_1 M_1}\right)\left(b_1+a_1 e^{\lambda_1 M_1}\right)}{\left(\lambda_1+Q_1\right)M_1} \\
             & = \dfrac{-Q_1 \left(b_1+a_1 e^{\lambda_1 M_1}\right)}{\left(\lambda_1+Q_1\right)M_1 b_1\lambda_1} \\
             & = \dfrac{-Q_1 e^{Q_1 M_1} \left(b_1+a_1 e^{\lambda_1 M_1}\right)}{M_1 \left(b_1\lambda_1 + Q_1\right) \left(\lambda_1+Q_1\right)} \\
             & = \dfrac{k_1 e^{Q_1 M_1} \left(b_1+a_1 e^{\lambda_1 M_1}\right)}{\lambda_1+Q_1}.
    \end{split}
\end{equation}
Now, the validation of Eq. (\ref{APPENDIX C 12b}) in  Eq. (\ref{APPENDIX C 28}) indicates that the unique solution $g_{11}(x)$ in Eq. (\ref{APPENDIX C 25}) for Eq. (\ref{APPENDIX C 16}) and the unique solution $g_{12}(x) = k_1 e^{Q_1 x}$ for Eq. (\ref{APPENDIX C 7}) are obtained.

%
%




\end{document}